\documentclass[12pt,reqno]{amsart}
\usepackage[left=1.205in,right=1.205in,top=1.205in,bottom=1.205in]{geometry}
\usepackage{graphicx}
\usepackage{color}
\usepackage[dvipsnames]{xcolor}
\usepackage{amsmath, amsthm, amssymb, amsfonts}
\usepackage{bigints}
\usepackage{natbib}
\usepackage{har2nat}
\usepackage{tabularx}
\usepackage{caption}
\usepackage{parskip}
\usepackage[british]{babel}
\usepackage{mathtools}
\usepackage{xfrac}
\usepackage[colorlinks=true,allcolors=Blue,linkcolor=Blue]{hyperref}
\usepackage{dsfont}
\usepackage{hyperref}
\usepackage{setspace}
\usepackage{tikz}
\usepackage{caption}
\usepackage{pgfplots}
\usepackage{float}
\restylefloat{table}
\usepgfplotslibrary{fillbetween}
\usetikzlibrary{patterns}
\pgfplotsset{compat=newest}
\usetikzlibrary{decorations.pathreplacing,angles,quotes}
\usepackage{graphicx}
\usepackage{subcaption}

\usepackage{threeparttable}

\tikzset{
	declare function={
		normcdf(\x,\m,\s)=1/(1 + exp(-0.07056*((\x-\m)/\s)^3 - 1.5976*(\x-\m)/\s));
	}
}

\newtheorem{lemma}{Lemma}

\newtheorem{proposition}{Proposition}

\theoremstyle{definition}
\newtheorem{remark}{Remark}

\newtheorem{assumption}{Assumption}

\newcommand{\ul}{\underline}
\newcommand{\ol}{\overline}

\newcommand{\bdis}{\begin{displaymath}}
\newcommand{\edis}{\end{displaymath}}
\newcommand{\beq}{\begin{equation}}
\newcommand{\eeq}{\end{equation}}
\newcommand{\bea}{\begin{eqnarray*}}
\newcommand{\eea}{\end{eqnarray*}}
\newcommand{\bean}{\begin{eqnarray}}
\newcommand{\eean}{\end{eqnarray}}

\newcommand{\R}{\mathbb{R}}
\newcommand{\E}{\mathbb{E}}

\DeclareMathOperator\supp{supp}

\newcommand{\1}{\mbox{\bf 1}}

\def\Pr{\mathop{\rm Pr}\nolimits}

\def \a {\alpha }
\def \b {\beta }
\def \c {\gamma }
\def \wh {\widehat }

\begin{document}
\begin{titlepage}
\title{%
The Economics of Partisan Gerrymandering}
\author[\uppercase{Kolotilin and Wolitzky}]{\larger \textsc{Anton Kolotilin
and Alexander Wolitzky}}
\date{\today. \emph{First version:} 17th September 2020.}
\thanks{ \\
\textit{Kolotilin}: School of Economics, UNSW Business School.
\textit{Wolitzky}: Department of Economics, MIT.\\
\ \\
We thank Nikhil Agarwal, Garance Genicot, Ben Golub, Richard Holden, Gary King, Hongyi Li, Nolan McCarty, Stephen Morris, Ben Olken, and Ken Shotts, as well as seminar and conference participants at ASSA, Harvard, MIT, NBER, Peking, Penn State, Rochester, Stanford, and Warwick for helpful comments and suggestions. We thank Eitan Sapiro-Gheiler and Nancy Wang for excellent research assistance. Anton Kolotilin gratefully acknowledges support from the
Australian Research Council Discovery Early Career Research Award
DE160100964 and from MIT Sloan's Program on Innovation in Markets and
Organizations. Alexander Wolitzky gratefully acknowledges support from NSF CAREER Award 1555071 and Sloan Foundation Fellowship 2017-9633.
}

\renewcommand{\baselinestretch}{1}
\maketitle

\begin{abstract}
We study the problem of a partisan gerrymanderer who assigns voters to equipopulous districts so as to maximize his party's expected seat share. The designer faces both aggregate uncertainty (how many votes his party will receive) and idiosyncratic, voter-level uncertainty (which voters will vote for his party). We argue that \emph{pack-and-pair districting}, where weaker districts are ``packed'' with a single type of voter, while stronger districts contain two voter types, is typically optimal for the gerrymanderer. The optimal form of pack-and-pair districting depends on the relative amounts of aggregate and idiosyncratic uncertainty. When idiosyncratic uncertainty dominates, it is optimal to pack opposing voters and pair more favorable voters; this plan resembles traditional ``packing-and-cracking.'' When aggregate uncertainty dominates, it is optimal to pack moderate voters and pair extreme voters; this ``matching slices'' plan has received some attention in the literature. Estimating the model using precinct-level returns from recent US House elections indicates that, in practice, idiosyncratic uncertainty dominates and packing opponents is optimal; moreover, traditional pack-and-crack districting is approximately optimal. %
We discuss implications for redistricting reform and political polarization. Methodologically, we exploit a formal connection between gerrymandering---partitioning voters into districts---and information design---partitioning states of the world into signals.
\ \\

\noindent\emph{JEL\ Classification:}\ C78, D72, D82 %


\ 
\end{abstract}

\thispagestyle{empty}
\let\MakeUppercase\relax %
\end{titlepage}

\newpage

\section{Introduction} \label{s:intro}
\setcounter{page}{1}

Legislative district boundaries are drawn by political partisans under many electoral systems \citep{Bickerstaff}. In the United States, the importance of districting has accelerated with the rise of computer-assisted districting \citep{Newkirk}, together with intense partisan efforts to gain and exploit control of the districting process. These trends culminated in ``The Great Gerrymander of 2012''  \citep{McGhee}, where the Republican party's Redistricting Majority Project (REDMAP), having previously targeted state-level elections that would give Republicans control of redistricting, aggressively redistricted several states, including Michigan, Ohio, Pennsylvania, and Wisconsin. The resulting districting plans are widely viewed as contributing to the outcome of the 2012 general election, where Republican congressional candidates won a 33-seat majority in the House of Representatives with 49.4\% of the two-party vote \citep{McGann}. In light of these developments---along with the Supreme Court ruling in \emph{Rucho v. Common Cause} (2019) that partisan gerrymanders are not judiciable in federal court, and the continued prominence of gerrymandering in the 2020 US redistricting cycle \citep{Rakich}---partisan gerrymandering looks likely to remain an important feature of American politics for some time.

This paper studies the problem of a partisan gerrymanderer (the ``designer'') who assigns voters to a large number of equipopulous districts so as to maximize his party's expected seat share.\footnote{Of course, studying this problem does not endorse gerrymandering, any more than studying monopolistic behavior endorses monopoly.} This problem approximates the one facing many partisan gerrymanderers in the United States. In particular, the constraint that districts must be equipopulous is crucial and is strictly enforced by law.\footnote{In \emph{Karcher v. Daggett} (1983), the Supreme Court rejected a districting plan in New Jersey with less than a 1\% deviation from population equality, finding that “there are no \emph{de minimus} population variations, which could practically be avoided, but which nonetheless meet the standard of Article I, Section 2 [of the U.S. Constitution] without justification.”} In practice, gerrymanderers also face other significant constraints, such as the federal requirements that districts are contiguous and do not discriminate on the basis of race, and various state-level restrictions, such as ``compactness'' requirements, requirements to respect political sub-divisions such as county lines, requirements to represent racial or ethnic groups or other communities of interest, and so on. While these complex additional constraints are important in some cases, we believe that often they are not as binding as they might seem, and also that they are more productively considered on a case-by-case basis rather than as part of a general theoretical analysis.\footnote{See \citet{FH} for more discussion of these constraints. For example, contiguity is not as severe a constraint as it might seem, because contiguous districts can have extremely irregular shapes.} We therefore follow much of the literature (discussed below) in focusing on the simpler problem with only the equipopulation constraint.

When the designer has perfect information, it is well-known that the solution to this problem is \emph{pack-and-crack}: if the designer's party is supported by a minority of voters of size $m <1/2$, he ``packs'' $1-2m$ opposing voters in districts where he receives zero votes, and ``cracks'' the remaining $2m$ voters in districts which he wins with 50\% of the vote.\footnote{If the designer has majority support, he can win all the districts.} We instead consider the more general and realistic case where the designer must allocate a variety of types of voters (or, more realistically, groups of voters such as census blocks or precincts) under uncertainty. The goal of this paper is to characterize optimal partisan gerrymandering in this setting, to compare optimal gerrymandering with simple and realistic forms of packing-and-cracking, and to draw some implications for broader legal and political economy issues.

In outline, our model and results are as follows. We assume that the designer faces both aggregate uncertainty (how many votes his party will receive) and idiosyncratic, voter-level uncertainty (which voters will vote for his party). Aggregate uncertainty is parameterized by a one-dimensional aggregate shock, while voters are parameterized by a one-dimensional type that determines a voter's probability of voting for the designer's party for each value of the aggregate shock. We focus on the case where the aggregate shock is unimodal and where moderate voters are ``swingier'' than more extreme voters, in that their vote probabilities swing more with the aggregate shock. In this case, we argue that a class of districting plans that we call \emph{pack-and-pair}---which generalize pack-and-crack---are typically optimal for the designer. Under pack-and-pair districting, the designer creates weaker districts that are packed with a single type of voter (which are analogous to the packed districts under pack-and-crack), and stronger districts that contain two voter types (which are analogous to the cracked districts under pack-and-crack).

We further show that the optimal form of pack-and-pair districting depends on the relative amounts of aggregate and idiosyncratic uncertainty. When idiosyncratic uncertainty dominates, it is optimal to pack opposing voters and pair more favorable voters. This \emph{pack-opponents-and-pair} plan (henceforth, POP) resembles traditional packing-and-cracking. POP also resembles the ``$p$-segregation'' plan introduced by \citet{GP}, where opposing voters are segregated and more favorable voters are all pooled together, rather than being paired as they are under POP. When instead aggregate uncertainty dominates, it is optimal to pack moderate voters and pair extreme voters. This \emph{pack-moderates-and-pair} plan (henceforth, PMP) was proposed under the name  ``matching slices'' by \citet{FH} and was applied to redistricting law by \citet{CH}. The pack-and-pair class thus nests the main districting plans proposed in the literature. Our primary theoretical contribution is identifying this class and showing that the optimal plan within this class is determined by the relative amounts of aggregate and idiosyncratic uncertainty.

A rough intuition for these results is that when idiosyncratic uncertainty dominates, the probability that the designer wins a district is approximately determined by the mean voter type in the district, as in probabilistic voting models with partisan taste shocks (e.g., \citealt{Hinich}, \citealt{LindbeckWeibull}). With a unimodal aggregate shock, the distribution of district means is then optimized by segregating opposing voters and pooling more favorable voters, as in $p$-segregation. When instead aggregate uncertainty dominates, the probability that the designer wins a district is approximately determined by the median voter type in the district, as in probabilistic voting models with an uncertain median bliss point (e.g., \citealt{Wittman}, \citealt{Calvert}). The distribution of district medians is then optimized by pairing above-population-median and below-population-median voter types, as in matching slices. However, the optimal plans we identify (POP and PMP) are somewhat more intricate than $p$-segregation and the simple form of matching slices emphasized by \citet{FH}: POP pairs favorable voters, rather than pooling them as in $p$-segregation; and PMP segregates an interval of intermediate voter types, rather than pairing all types as in the simplest form of matching slices.

As we discuss in Section \ref{s:discussion}, whether optimal districting takes the form of POP or PMP has significant implications for several political and legal issues surrounding redistricting, including redistricting reform and intra- and inter-district political polarization (see also \citealt{CH}). It is therefore important to understand whether idiosyncratic or aggregate uncertainty is larger in practice. We answer this question using precinct-level returns from the 2016, 2018, and 2020 US House elections. The data clearly show that idiosyncratic uncertainty is much larger than aggregate uncertainty. Intuitively, this finding results from the simple observation that, in practice, most precinct vote splits are much closer to 50-50 (the vote split under high idiosyncratic uncertainty) than 100-0 or 0-100 (the vote splits under high aggregate uncertainty).\footnote{This observation also implies that models with only two types of voters or precincts (e.g., \citealt{OG}) cannot closely approximate the problem facing actual gerrymanderers, who must decide how to allocate many different types of precincts.} We therefore expect that, in practice, optimal districting takes the form of POP. We also note, however, that the optimal POP plan is close to $p$-segregation under our estimated parameters. Thus, simple $p$-segregation plans are likely approximately optimal in practice. This finding helps explain why actual gerrymandering usually resembles $p$-segregation---or an even simpler form of pack-and-crack, where unfavorable voters are pooled rather than segregated---instead of a more complicated plan like POP.

Methodologically, we establish a formal connection between gerrymandering---partitioning voters into districts---and information design---partitioning states of the world into signals. The partisan gerrymandering problem we study is mathematically equivalent to a general Bayesian persuasion problem with a one-dimensional state, a one-dimensional action for the receiver, and state-independent sender preferences. Most of our results are novel in the context of this persuasion problem. This paper thus directly contributes to information design as well as gerrymandering; more importantly, we establish a strong connection between these two topics.\footnote{Contemporaneous papers by \citet{Lagarde} and \citet{GPS} also emphasize connections between gerrymandering and information design, albeit in less general models. \citeauthor{Lagarde} assume two voter types, as in \citet{OG}; \citeauthor{GPS} assume no aggregate uncertainty. The closest paper in the persuasion literature is our companion paper, \citet{KCW}, which we discuss later on.} 

\subsection{Related Literature} The most related prior papers on optimal partisan gerrymandering are \citet{OG}, \citet{FH}, and \citet{GP}. \citeauthor{OG}'s model is equivalent to the special case of our model with two voter types. \citeauthor{GP} consider competition between two designers who each control districting in some area and aim to win a majority of seats.\footnote{\citet{FH20} study designer competition in the model of their \citeyear*{FH} paper.} A simplified version of their model with a single designer is equivalent to the special case of our model where vote swings are linear in voter types; we discuss this special case in Section \ref{s:linear}. \citeauthor{FH} consider essentially the same model as we do (and in particular allow non-linear swings), but their main results concern the special case where aggregate uncertainty is much larger than idiosyncratic uncertainty. In contrast, we do not restrict the relative amounts of aggregate and idiosyncratic uncertainty, and we show empirically that the practically relevant case is that where idiosyncratic uncertainty dominates (i.e., the opposite of the case emphasized by \citeauthor{FH}).

The broader literature on gerrymandering and redistricting addresses a wide range of issues, including geographic constraints on gerrymandering \citep{Sherstyuk,Shotts01,PT}, gerrymandering with heterogeneous voter turnout \citep{Bouton}, socially optimal districting \citep{Gilligan,CK,Bracco}, measuring district compactness \citep{CM,Fryer,Ely}, the interaction of redistricting and policy choices \citep{Shotts02,BesleyPreston}, measuring gerrymandering \citep{GrofKing,MG,SM,Duchin,GPS}, and assessing the consequences of redistricting \citep[among many:][]{Gelman,MPR,HM,Jeong}. As the partisan gerrymandering problem interacts with many of these issues, our analysis may facilitate future research in these areas.

\subsection{Outline} The paper is organized as follows: Section \ref{s:model} presents the model. Section \ref{s:benchmark} analyzes some benchmark cases. Section \ref{s:nonlinear} contains our main theoretical and numerical results. Section \ref{s:estimation} contains our empirical results. Section \ref{s:discussion} discusses policy implications of our results. Section \ref{s:conclusion} concludes. All proofs are deferred to the appendix.

\section{Model}\label{s:model}

We consider a standard electoral model with one-dimensional voter types (parameterizing a voter's probability of voting for the designer's party) and one-dimensional aggregate uncertainty (parameterizing the designer's aggregate vote share).

\textbf{Voters and Vote Shares.} There is a continuum of voters. Each voter has a type $s\in [\ul s,\ol s]$, which is observed by the designer.\footnote{In our empirical implementation, $s$ will correspond to the precinct the voter lives in.} The population distribution of voter types is denoted by $F$. The aggregate shock is denoted by $r\in \R$; its distribution is denoted by $G$. We assume that $F$ and $G$ are sufficiently smooth and that the corresponding densities $f$ and $g$ are strictly positive.\footnote{It suffices that distributions $F$, $G$, and $Q$ (defined below) are four-times differentiable. We also consider discrete distributions in some benchmark cases.}

The share of type-$s$ voters who vote for the designer when the aggregate shock takes value $r$ is deterministic and is denoted by $v(s,r)\in[0,1]$.\footnote{In our empirical implementation, $v(s,r)$ will correspond to the designer's vote share in precinct $s$ given shock $r$.} The function $v(s,r)$ plays a key role in our analysis. We assume that $v(s,r)$ is strictly increasing in $s$ and strictly decreasing in $r$. Thus, higher voter types are stronger supporters of the designer (i.e., they vote for him with higher probability for every $r$), and higher aggregate shocks are worse for the designer (i.e., they reduce the probability that each voter type votes for him). The model thus lets different voter types ``swing'' by different amounts in response to an aggregate shock, but it does assume that all types swing in the same direction. We also impose the technical assumptions that $v(s,r)$ is four-times differentiable and satisfies $\lim_{r\to \infty}v(s,r)=0$ and $\lim_{r\to -\infty}v(s,r)=1$ for all $s$.

An interpretation of the vote share function $v(s,r)$ is that each voter is hit by an idiosyncratic ``taste shock'' $t \in \R$ and votes for the designer if and only if $$s-r-t \geq 0.$$ With this interpretation, when the taste shock distribution is $Q$, we have \[v(s,r) = Q(s-r)\; \; \text{for all } (s,r).\] Mathematically, this ``additive taste shock'' case arises when the function $v(s,r)$ is translation-invariant: i.e., depends only on the difference $s-r$. In this case, the model is parameterized by three distributions: $F$, $G$, and $Q$. However, scaling $s$, $r$, and $t$ by the same constant leaves the model unchanged, so we can normalize the variance of one of these three variables to $1$. We will thus assume, without loss, that the variance of $t$ is $1$.\footnote{Outside of the benchmark case considered in Section \ref{s:noindiv}, where $Q$ is degenerate.}

The designer thus faces two kinds of uncertainty: aggregate uncertainty (captured by $r$) and idiosyncratic, voter-level uncertainty (captured by $t$, or more generally by the extent to which $v(s,r)$ lies away from the extremes of $0$ and $1$). Many of our results will involve comparing the ``amount'' of each kind of uncertainty.

\textbf{Districting Plans.} The designer allocates voters among a continuum of equipopulous districts based on their types $s$, and thus determines the distribution $P$ of $s$ in each district.\footnote{Since districting plans in the US are drawn at the state level, our continuum model implicitly assumes that each state contains a large number of districts. Obviously, this is a better approximation for state legislative districts and for congressional districts in large states than it is for congressional districts in small states. Introducing integer constraints on the number of districts, while interesting and realistic, would substantially complicate the analysis and would risk obscuring our main insights.} A district is characterized by the distribution $P$ of voter types $s$ it contains. Thus, a \emph{districting plan}---which specifies the measure of districts with each voter-type distribution $P$---is a distribution $\mathcal H$ over distributions $P$ of $s$, such that the population distribution of $s$ is given by $F$: that is, $\mathcal H \in \Delta \Delta \mathbb [\ul s,\ol s]$ and
\begin{equation*}\label{mps}
\int P(s)d\mathcal{H}(P)=F(s) \; \; \text{for all } s.
\end{equation*}
For example, under \emph{uniform districting}, where all districts are the same, $\mathcal{H}$ assigns probability $1$ to $P=F$. In the opposite extreme case of $\emph{segregation}$, where each district consists entirely of one type of voter, every distribution $P$ in the support of $\mathcal{H}$ takes the form $P=\delta_{s}$ for some $s\in [\ul s,\ol s]$, where $\delta_{s}$ denotes the degenerate distribution on voter type $s$.

\textbf{Designer's Problem.} The designer wins a district iff he receives a majority of the district vote. Thus, the designer wins a district with voter type distribution $P$ (henceforth, ``district $P$'') iff $r$ satisfies $\int v(s,r) dP(s)\geq 1/2$. Since $v(s,r)$ is decreasing in $r$, the designer wins district $P$ iff
\[
r \leq r^*(P):=\left\{r:\int v(s,r)dP(s)= \frac12\right\}.
\]

We say that a district $P'$ is \emph{weaker} than another district $P$ if $r^*(P')<r^*(P)$. Note that, whenever the designer wins a district $P$, he also wins all weaker districts $P'$. Our model thus reflects what \citeauthor{GrofKing} (\citeyear*{GrofKing}, p. 12) call ``a key empirical generalization that applies to all elections in the U.S. and most other democracies: the statewide or nationwide swing in elections is highly variable and difficult to predict, but the approximate rank order of districts is highly regular and stable.''%

We assume that the designer maximizes his party's expected seat share.\footnote{See Section \ref{s:conclusion} and \citet{KW} for discussion of more general designer objectives.} Thus, the designer's problem is 
\begin{gather*} \label{persuasion}
	\max_{\mathcal H \in \Delta \Delta \mathbb [\ul s,\ol s]} \int G(r^{*}(P)) d\mathcal{H}(P) \\
	\text{s.t.} \; \int Pd\mathcal{H}(P)=F. \nonumber
\end{gather*}

This problem nests the partisan gerrymandering problems of \citet{OG}, \citet{FH}, and (with a single designer) \citet{GP}.\footnote{\citet{GP} consider a majoritarian objective with district-level uncertainty in addition to aggregate uncertainty. However, after conditioning on the pivotal value of the aggregate shock, district-level uncertainty in \citeauthor{GP} plays the same role as aggregate uncertainty in our model.} It is also equivalent to a Bayesian persuasion problem, where the designer splits a prior distribution $F$ into posterior distributions $P$, and obtains utility $G(r^{*}(P))$ from inducing posterior $P$.\footnote{Specifically, the designer's problem is equivalent to the \emph{state-independent sender} case of the persuasion problem studied in \citet{KCW}, which specializes the general Bayesian persuasion problem of \citet{KG} by assuming that the state and the receiver's action are one-dimensional, the receiver's utility is supermodular and concave in his action, and the sender's utility is independent of the state and increasing in the receiver's action. In the gerrymandering context, state-independent sender preferences reflect the fact that the designer cares only about how many districts he wins and not directly about the composition of these districts.}

\section{Benchmark Cases} \label{s:benchmark}

We first consider four benchmark cases:
\begin{enumerate}
\item There is no uncertainty.
\item There is idiosyncratic uncertainty but no aggregate uncertainty.
\item There is aggregate uncertainty but no idiosyncratic uncertainty.
\item Both kinds of uncertainty are present, but swings are linear in voter types.
\end{enumerate}

These cases illustrate the key forces in the model and set up our main analysis. The benchmark cases with only one kind of uncertainty are much more tractable than the general case with both kinds, but they give a good indication of the form of optimal districting plans when both kinds of uncertainty are present but one kind is much ``larger'' than the other. We will see that this case is relevant in practice, where idiosyncratic uncertainty is much larger than aggregate uncertainty. Similarly, the linear swing case is very tractable and is a good guide to the more realistic case where swings deviate from linearity systematically but by a relatively small amount.

\subsection{Perfect Information: Pack-and-Crack}\label{s:classic}

With perfect information, optimal gerrymandering takes a simple and well-known form.

\begin{proposition}\label{packandcrack}
Assume there is no uncertainty: there exists $r^0$ such that $r=r^0$ with certainty, and $v(s,r^0)=\1\{s\geq r^0\}$ for all $s$. Denote the fraction of the designer's ``supporters'' by $m=1-F(r^0)$.
\begin{enumerate}
\item If $m\geq 1/2$, a districting plan is optimal iff it creates measure $1$ of districts where $\Pr_{P}(s\geq r^0)\geq 1/2$. Under such a plan, the designer wins all districts.

\item If $m< 1/2$, a districting plan is optimal iff it creates measure $2m$ of ``cracked'' districts where $\Pr_{P}(s\geq r^0)=\Pr_{P}(s<r^0)=1/2$ and measure $1-2m$ of ``packed'' districts where $\Pr_{P}(s<r^0)=1$. Under such a plan, the designer wins the cracked districts. 
\end{enumerate}
\end{proposition}

Case (1) says that a designer with majority support wins all the districts (e.g., with uniform districting). Case (2) says that a designer with minority support $m<1/2$ wins $2m$ districts with 50\% of the vote, and gets zero votes in the remaining $1-2m$ districts. We call any optimal plan in case (2) \emph{pack-and-crack}.

When $m<1/2$ and voter types are continuous, there are many pack-and-crack plans. For example, some types of supporters can be assigned to only a subset of cracked districts, and some types of opponents can be assigned only to packed districts. This seemingly pedantic point will become important once we introduce uncertainty, because optimal plans under a small amount of uncertainty will approximate some but not all pack-and-crack plans.

\begin{figure}[t]
\centering

\begin{tabular}{@{}l@{}}
\hspace{0.05cm}
\begin{tikzpicture}[scale = .65]
	
	\pgfmathsetmacro{\l}{10}
	\pgfmathsetmacro{\x}{3.5}
	
	\fill	[fill = blue, opacity = 0.3]	(0, 0)	--	(\l-2*\x, 0) -- (\l-2*\x, 0.3) -- (0, 0.3);
	\fill	[fill = red, opacity = 0.3]	(\l-2*\x, 0)	--	(\l, 0) -- (\l, 0.3) -- (\l-2*\x, 0.3);
	
	\draw	[ultra thick, smooth, blue]	(0, 0) -- (\l-\x, 0);
	\draw	[ultra thick, smooth, red]	(\l-\x, 0) -- (\l, 0);
	
	\draw	[thick, smooth, black]	(0, -0.15) -- (0, 0.15);
	\draw	[thick, smooth, black]	(\l-2*\x, -0.15) -- (\l-2*\x, 0.15);
	\draw	[thick, smooth, black]	(\l-\x, -0.15) -- (\l-\x, 0.15);
	\draw	[thick, smooth, black]	(\l, -0.15) -- (\l, 0.15);
	
	\node[label={$1-2x^0$}]		at (\l/2-\x, -1.25)			{};
	\node[label={$x^0$}]		at (\l-2*\x+\x/2, -1.25)	{};
	\node[label={$x^0$}]		at (\l - \x/2, -1.25)		{};
	
	\draw [decorate,decoration={brace,amplitude=15pt,raise=3ex}]
	(0,0) -- (\l-2*\x,0) node[midway,yshift=3.5em]{pooling};
	\draw [decorate,decoration={brace,amplitude=15pt,raise=3ex}]
	(\l-2*\x,0) -- (\l,0) node[midway,yshift=3.5em]{pooling};
	
	\node[label={(a) traditional pack-and-crack}]	at (\l/2, -2.5)	{};
	\node[label={}]	at (\l/2, -3.5)	{};
\end{tikzpicture}

\hspace{1cm}
\begin{tikzpicture}[scale = .65]
	
	\pgfmathsetmacro{\l}{10}
	\pgfmathsetmacro{\x}{3.5}
	
	\pattern	[pattern = north east lines, pattern color = blue, opacity = 0.7]	(0, 0)	--	(\l-2*\x, 0) -- (\l-2*\x, 0.3) -- (0, 0.3);
	\fill	[fill = red, opacity = 0.3]	(\l-2*\x, 0)	--	(\l, 0) -- (\l, 0.3) -- (\l-2*\x, 0.3);
	
	\draw	[ultra thick, smooth, blue]	(0, 0) -- (\l-\x, 0);
	\draw	[ultra thick, smooth, red]	(\l-\x, 0) -- (\l, 0);
	
	\draw	[thick, smooth, black]	(0, -0.15) -- (0, 0.15);
	\draw	[thick, smooth, black]	(\l-2*\x, -0.15) -- (\l-2*\x, 0.15);
	\draw	[thick, smooth, black]	(\l-\x, -0.15) -- (\l-\x, 0.15);
	\draw	[thick, smooth, black]	(\l, -0.15) -- (\l, 0.15);
		
	\node[label={$1-2x^0$}]		at (\l/2-\x, -1.25)			{};
	\node[label={$x^0$}]		at (\l-2*\x+\x/2, -1.25)	{};
	\node[label={$x^0$}]		at (\l-\x/2, -1.25)			{};
	
	\draw [decorate,decoration={brace,amplitude=15pt,raise=3ex}]
	(0,0) -- (\l-2*\x,0) node[midway,yshift=3.5em]{segregation};
	\draw [decorate,decoration={brace,amplitude=15pt,raise=3ex}]
	(\l-2*\x,0) -- (\l,0) node[midway,yshift=3.5em]{pooling};
	
	\node[label={(b) pack-opponents-and-pool}]	at (\l/2, -2.5)	{};
	\node[label={($p$-segregation)}]	at (\l/2, -3.5)	{};
\end{tikzpicture}

\vspace{0cm}
\\
\begin{tikzpicture}[scale = .65]
	
	\pgfmathsetmacro{\l}{10}
	\pgfmathsetmacro{\x}{3.5}
	
	\draw[thick, red] (\l, 0) arc (0:180:\x cm and 1.5cm);
	\draw[thick, red] (\l-\x/3, 0) arc (0:180:2/3*\x cm and 1cm);
	\draw[thick, red] (\l-2*\x/3, 0) arc (0:180:1/3*\x cm and 0.5cm);
	
	\pattern	[pattern = north east lines, pattern color = blue, opacity = 0.7]	(0, 0)	--	(\l-2*\x, 0) -- (\l-2*\x, 0.3) -- (0, 0.3);
	\pattern	[pattern = north east lines, pattern color = red, opacity = 0.7] 	(\l-\x, 0)	--	(\l-\x+0.1, 0) -- (\l-\x+0.1, 0.3) -- (\l-\x, 0.3);
	
	\draw	[ultra thick, smooth, blue]	(0, 0) -- (\l - \x, 0);
	\draw	[ultra thick, smooth, red]	(\l - \x, 0) -- (\l, 0);
	
	\draw	[thick, smooth, black]	(0, -0.15) -- (0, 0.15);
	\draw	[thick, smooth, black]	(\l - 2*\x, -0.15) -- (\l - 2*\x, 0.15);
	\draw	[thick, smooth, black]	(\l - \x, -0.15) -- (\l - \x, 0.15);
	\draw	[thick, smooth, black]	(\l, -0.15) -- (\l, 0.15);
	
	\node[label={$1-2x^0$}]		at (\l/2-\x, -1.25)			{};
	\node[label={$x^0$}]		at (\l-2*\x+\x/2, -1.25)	{};
	\node[label={$x^0$}]		at (\l-\x/2, -1.25)			{};
	
	\draw [decorate,decoration={brace,amplitude=15pt,raise=8ex}]
	(0,0) -- (\l-2*\x,0) node[midway,yshift=6em]{segregation};
	\draw [decorate,decoration={brace,amplitude=15pt,raise=8ex}]
	(\l-2*\x,0) -- (\l,0) node[midway,yshift=6em]{pairing};
	
	\node[label={(c) pack-opponents-and-pair (POP)}]	at (\l/2, -2.5)	{};
	\node[label={}]	at (\l/2, -3.5)	{};
\end{tikzpicture}

\hspace{1cm}
\begin{tikzpicture}[scale = .65]
	
	\pgfmathsetmacro{\l}{10}
	\pgfmathsetmacro{\x}{3.5}
	\pgfmathsetmacro{\a}{1/2*\l-1/3*\x}
	\pgfmathsetmacro{\b}{1/2*\l-2/3*\x}
	\pgfmathsetmacro{\c}{1/2*\l-\x}
	\pgfmathsetmacro{\d}{\l/4-\x/2}

	\draw[thick, red] (\l, 0) arc (0:180:1/2*\l cm and 2cm);
	\draw[thick, red] (\l-1/3*\x, 0) arc (0:180:\a cm and 1.5cm);
	\draw[thick, red] (\l-2/3*\x, 0) arc (0:180:\b cm and 1cm);
	\draw[ultra thick, red] (\l-\x, 0) arc (0:180:\c cm and 0.5cm);
	\draw[thick, blue] (\l-\x, 0) arc (0:180:\c cm and 0.5cm);
	\draw[thick, blue] (3/4*\l-\x/2, 0) arc (0:180:\d cm and 0.25cm);
	
	\pattern	[pattern = north east lines, pattern color = blue, opacity = 0.7]	(4/3*\x, 0)	--	(\l-4/3*\x, 0) -- (\l-4/3*\x, 0.3) -- (4/3*\x, 0.3);
	
	\draw	[very thick, smooth, blue]	(0, 0) -- (\l-\x, 0);
	\draw	[very thick, smooth, red]	(\l-\x, 0) -- (\l, 0);
	
	\draw	[thick, smooth, black]	(0, -0.15) -- (0, 0.15);
	\draw	[thick, smooth, black]	(\x, -0.15) -- (\x, 0.15);
	\draw	[thick, smooth, black]	(\l-\x, -0.15) -- (\l-\x, 0.15);
	\draw	[thick, smooth, black]	(\l, -0.15) -- (\l, 0.15);
	
	\node[label={$x^0$}]		at (\x/2, -1.25)		{};
	\node[label={$1 - 2x^0$}]	at (\l/2, -1.25)		{};
	\node[label={$x^0$}]		at (\l-\x/2, -1.25)		{};
	
	\draw [decorate,decoration={brace,amplitude=15pt,raise=8ex}]
	(0,0) -- (\l,0) node[midway,yshift=6em]{pairing};
	
	\node[label={(d) pack-moderates-and-pair (PMP)}]	at (\l/2, -2.5)	{};
	\node[label={(matching slices)}]	at (\l/2, -3.5)	{};
\end{tikzpicture}
\end{tabular}
\caption{Four Varieties of Pack-and-Crack}
\caption*{\emph{Notes:} In each panel, the horizontal axis is the interval of voter types, $s$, where red voters are supporters and blue voters are opponents. The designer wins red districts and loses blue ones. Solid shading indicates pooling; curved lines connecting two voter types indicate pairing; hatched shading indicates segregation.}
\label{f:4pc}
\end{figure}

Figure \ref{f:4pc} illustrates four pack-and-crack plans that play important roles in our analysis. Panel (a) is what we call \emph{traditional pack-and-crack}: the strongest opposing voters are pooled in one type of district, while the remaining voters (a mix of supporters and opponents) are pooled in another type of district. Panel (b) is the same, except now each strong opposing type is segregated in a distinct, homogeneous district. We call this plan \emph{pack-opponents-and-pool}. This plan was previously studied by \citet{GP}, who called it ``$p$-segregation.'' Panel (c) is the same as Panel (b), except now favorable voter types are matched in a negatively assortative manner to form distinct districts. We call this plan \emph{pack-opponents-and-pair}, or POP. This plan plays a central role in our analysis, as we will see that it is optimal for realistic parameter values; however, we will also see that the simpler traditional pack-and-crack and pack-opponents-and-pool plans are approximately optimal for the same parameters.

Finally, we call the plan in Panel (d)---where extreme voter types are matched in a negatively assortative manner, and intermediate voter types are segregated---\emph{pack-moderates-and-pair}, or PMP. This plan was previously studied by \citet{FH}, who called it ``matching slices.''\footnote{\citeauthor{FH} did not emphasize the possibility of segregating a non-trivial interval of intermediate voter types under matching slices, but their results allow this possibility, and we will see that this is actually the typical case.} We also refer to the extreme form of PMP where the segregation region is degenerate, so that only a single voter type is segregated, as \emph{negative assortative districting}.

\subsection{No Aggregate Uncertainty}\label{s:noaggr}

We next consider the case with idiosyncratic uncertainty but no aggregate uncertainty. As we will see, this case is fairly realistic, as empirically idiosyncratic uncertainty is much larger than aggregate uncertainty.

\begin{proposition}\label{nounc}
Assume there is no aggregate uncertainty: there exists $r^0$ such that $r=r^0$ with certainty.
\begin{enumerate}
\item If $\int v(s,r^0)dF(s)\geq 1/2$, a districting plan is optimal iff it creates measure $1$ of districts where $\int v(s,r^0)dP(s)\geq 1/2$. Under such a plan, the designer wins all districts. 

\item If $\int v(s,r^0)dF(s)< 1/2$, let $s^*$ satisfy $\int ^{\ol s}_{s^*} (v(s,r^0)-1/2) d F(s) =0.$ A districting plan is optimal iff it creates measure $1-F(s^*)$ of cracked districts where $\Pr_{P}(s\geq s^*)=1$ and $\int ^{\ol s}_{s^*} v(s,r^0) d P(s) =1/2$, and measure $F(s^*)$ of packed districts where $\Pr_{P}(s<s^*)=1$. Under such a plan, the designer wins the cracked districts.
\end{enumerate}
\end{proposition}

In case (1), the designer wins all districts under uniform districting. In case (2), the designer assigns all voter types $s>s^*$ to cracked districts that he wins with exactly 50\% of the vote, and packs the remaining voters arbitrarily. The intuition is that the designer wins a district iff the mean vote share $v(s,r^0)$ among voters in the district exceeds 50\%, so to win as many districts as possible the designer assigns only voter types above $s^*$ to cracked districts. This plan approximates the pack-and-crack vote share pattern as closely as possible, given the uncertainty facing the designer.

The optimal plans in Proposition \ref{nounc} coincide with the subset of optimal perfect-information plans that pack opponents (e.g., the plans in Figure \ref{f:4pc}(a)--(c)). Hence, pack-and-crack plans that pack opponents can be optimal with idiosyncratic uncertainty but no aggregate uncertainty, but plans that pack moderates (e.g., PMP) cannot. In Sections \ref{s:nonlinear} and \ref{s:estimation}, we will see that idiosyncratic uncertainty dominates aggregate uncertainty in practice. Hence, any optimal plan in Proposition \ref{nounc}---for example, traditional pack-and-crack---will prove to be approximately optimal for realistic parameters.

\subsection{No Idiosyncratic Uncertainty}\label{s:noindiv}

We now turn to the case with aggregate uncertainty but no idiosyncratic uncertainty.

\begin{proposition}\label{nohet}
Assume there is no idiosyncratic uncertainty: $v(s,r)=\1\{s\geq r\}$ for all $(s,r)$. Denote the population median voter type by $s^{m} = F^{-1}(1/2)$. A districting plan is optimal iff for $\mathcal H-$almost every district $P\in \supp(\mathcal H)$ there exists a voter type $s^{P} \geq s^{m}$ such that $\Pr_P(s=s^P)=\Pr_P(s< s^m)=1/2$. Under such a plan, the designer wins district $P$ iff $r\leq s^P$.
\end{proposition}

That is, for each voter type $s$ above the population median, the designer creates a district consisting of 50\% voters with this type and 50\% voters with below-median types. Note that, for every realization of aggregate uncertainty $r\in (\ul s,\ol s)$,  the designer wins some districts with exactly 50\% of the vote and wins zero votes in all other districts. This is precisely the pack-and-crack vote share pattern. 

The intuition for Proposition \ref{nohet} is easy to see with a finite number $N$ of districts. With no idiosyncratic uncertainty, the probability that the designer wins a given district is determined by the median voter type in that district. The strongest district the designer can possibly create is formed by combining the $1/(2N)$ highest voter types with any other voters: that is, it is impossible to create a district where the median voter is above the $1-1/(2N)$ quantile of the population distribution. Similarly, it is impossible to create $n$ districts where the median voter is everywhere above the $1-n/(2N)$ quantile of the population distribution. But, by creating districts one at time by always combining the $1/(2N)$ highest remaining voters with $1/(2N)$ below-median voters, the designer ensures that the median voter in the $n^{\text{th}}$ strongest district is exactly the $1-n/(2N)$ quantile. So this plan is optimal.

The optimal plans in Proposition \ref{nohet} are a subset of optimal perfect-information plans. 
For example, the PMP plan in Figure \ref{f:4pc}(d) remains optimal when $v(s,r)=\1\{s\geq r\}$ but $r$ is not degenerate, while the plans in Figures \ref{f:4pc}(a)--(c) that pack opponents are not optimal in this setting. This result is consistent with \citet{FH}, who show that matching slices is optimal when idiosyncratic uncertainty is sufficiently small, under some additional assumptions which we discuss in Section \ref{s:sdd}.\footnote{Note that in every optimal plan in Proposition \ref{nohet}, all voters with the highest type $s$ are assigned to the same district: in \citeauthor{FH}'s words, ``one's most ardent supporters should be grouped together.'' This is what \citeauthor{FH} mean when they write that ``cracking is never optimal'' and summarize their findings as ``sometimes pack, but never crack.''}

\subsection{Linear Swing} \label{s:linear}

Our last benchmark case is when vote shares and swings are linear in voter types. There are two equivalent ways to define this case. The simplest definition is that \emph{vote shares $v(s,r)$ are linear in $s$}:
\[ v(s,r)=\frac{\ol s-s}{\ol s-\ul s}v(\ul s,r)+\frac{s-\ul s}{\ol s-\ul s}v(\ol s,r)\quad \text{for all $(s,r)$}.\]
An alternative, equivalent definition is that vote swings are linear in $s$. To state this definition, first define the \emph{swing} of a voter type $s$ when the aggregate shock changes from $r'$ to $r$ by \[\Delta_{s}^{r,r'}=v(s,r)-v(s,r').\]
We then say that \emph{swings $\Delta_{s}^{r,r'}$ are linear in $s$} if 
\begin{align*}
\Delta_{s}^{r,r'} &= \rho(s)\Delta_{\ul s}^{r,r'}+(1-\rho(s))\Delta_{\ol s}^{r,r'} \quad \text{for all $(s,r,r')$,} \quad \text{where} \\
\rho(s)&=\frac{v(\ol s,r)-v(s,r)}{v(\ol s,r)-v(\ul s,r)}=\frac{v(\ol s,r')-v(s,r')}{v(\ol s,r')-v(\ul s,r')}.
\end{align*}
It is easy to see that, up to a rescaling of $s$, vote shares are linear iff swings are linear.

The linear case nests the \emph{uniform swing case} where $\Delta_{s}^{r,r'}$ is independent of $s$ (for each $r,r'$), so the aggregate shock shifts the vote share equally for all voter types. Political scientists often assume uniform swing to study how a given districting plan would perform under different electoral outcomes.\footnote{See, e.g., \citet{KKR} for a recent discussion of this methodology.} The linear case also nests the case where voter types are binary (i.e., $\supp(F) = \{\ul s,\ol s\}$), as well as the no-aggregate-uncertainty case considered in Section \ref{s:noaggr}. However, the no-idiosyncratic-uncertainty case considered in Section \ref{s:noindiv} cannot be linear, unless voter types are binary.

The key simplification afforded by linearity is that the threshold shock $r^*(P)$ for winning a district $P$ depends only on the district mean voter type $x=\E_{P}[s]$. Under linearity, the designer thus effectively chooses a distribution $H(x)$ of mean types $x$, rather than a distribution $\mathcal{H}(P)$ of distributions of types $P$. With this formulation, the constraint $\int Pd\mathcal{H}(P)=F$ simplifies to the requirement that $H$ is a mean-preserving contraction of $F$, which we denote by $F \succsim H$.\footnote{One way to see this is by analogy to statistics, where if a state $s$ is distributed according to $F$ then there exists an experiment such that the distribution of posterior expectations of $s$ is given by $H$ iff $H$ is a mean-preserving contraction of $F$ \citep[e.g.,][]{Blackwell,Kolotilin2017}. \label{f:MPS}} 

Slightly abusing notation, the designer wins districts with mean voter type at least $x$ iff $r \leq r^*(x)$. The probability of this event is 
\[U(x):=G(r^*(x)).\] 
We can interpret $U$ as the distribution of a re-scaled aggregate shock $z$ where the designer wins a district with mean type $x$ iff $z \leq x$. %
The designer's problem thus becomes
\begin{equation*}
\begin{gathered}
	\max_{H\in \Delta [\ul s,\ol s]} \int U (x) d H(x)\\
	\text{s.t. $F\succsim H$}. 
\end{gathered}
\end{equation*}

Clearly, uniform districting is optimal if $U$ is concave, and segregation is optimal if $U$ is convex. %
However, a more realistic assumption is that $U$ is \emph{strictly S-shaped}, so the marginal impact of replacing a less favorable voter with a more favorable one on the probability of winning a district is first increasing and then decreasing. Formally, this means that there is an inflection point $x^i\in [0,1]$ such that $U$ is strictly convex on $[0,x^i]$ and strictly concave on $[x^i,1]$; equivalently, the re-scaled aggregate shock $z$ is unimodal. %

We will see that $U$ being S-shaped is closely related to the optimality of pack-opponents-and-pool districting (i.e., $p$-segregation, see Figure \ref{f:4pc}(b)), where voter types below some cutoff $s^*$ are segregated, and voter types above $s^*$ are pooled in districts with mean voter type $x^* = \E_{F} [s|s\geq s^*]$. %
Under pack-opponents-and-pool districting with cutoff $s^*$ and pool mean $x^* = \E_{F} [s|s\geq s^*]$, the designer's expected seat share is
\[\int_{\ul s}^{s^*}U(x)dF(x)+U(x^*)(1-F(s^*)).\]
The best pack-opponents-and-pool plan is the one where $s^*$ is chosen to maximize this expectation. When the optimal value of $s^*$ is interior, it is characterized by the first-order condition
\begin{equation*} \label{FOC}
u(x^*)(x^* - s^*)=U(x^*)-U(s^*).\footnote{This equation is analogous to equation (12) of \citet{GP}.}
\end{equation*}
The intuition for this equation is that a marginal increase in $s^*$ increases the pool mean, which increases the designer's expected seat share by $u(x^*)(1-F(s^*))dx^{*}/ds^{*}=u(x^*)(x^{*}-s^*)f(s^*)$; but also decreases the mass of pooled voters, which decreases the designer's expected seat share by $(U(x^*)-U(s^*))f(s^*)$. The first-order condition equates the marginal benefit and marginal cost. See Figure \ref{f:sp}.

\begin{figure}[t]
\centering

\begin{tikzpicture}[scale = 1]
	
	\begin{axis}
		[axis x line = middle,
		axis y line = left,
		xmin = -3, xmax = 4,
		ymin = 0, ymax = 1.2,
		xlabel=$x$,
		ytick=1,
		xtick \empty,
		clip=false]
		
		\draw	[dashed, black]		(-1.6, 0) -- (-1.6, 0.0548);
		\draw	[dashed, black]		(-3, 0.0548) -- (-1.6, 0.0548);
		\draw	[dashed, black]		(-3, 0.758) -- (0.7, 0.758);
		\draw	[dashed, black]		(0.7,0) -- (0.7, 0.758);
		\draw	[ultra thin, solid, black]		(3, -0.015) -- (3, 0.015);
		
		\addplot [domain=-1.6:1.8, thick, dashed, red]	{0.307*x+(0.0548+1.6*0.307)};
		\addplot [domain=-3:3, thick, smooth, blue]			{normcdf(x,0,1)};
		
		\pattern	[pattern = north east lines, pattern color = blue, opacity = 0.7]	(-3, 0)	--	(-1.6, 0) -- (-1.6, 0.03) -- (-3,0.03);
		\fill	[fill = red, opacity = 0.3]	(-1.6, 0)	--	(3, 0) -- (3, 0.03) -- (-1.6, 0.03);
		
		\node[label={$\ul s$}] 		at (-3, -0.16) 		{};
		\node[label={$\ol s$}] 		at (3, -0.16) 		{};
		\node[label={${U}(s^*)$}] at (-3.6, -0.03) 	{};
		\node[label={${U}(x^*)$}] at (-3.6, 0.67) 		{};
		\node[label={$s^*$}] 	at (-1.6, -0.16) 	{};
		\node[label={$x^*$}] 	at (0.7, -0.16) 	{};
		\node[label={\textcolor{blue}{$U$}}] at (2.5, 0.85) {};
		
		\normalsize
		\draw [decorate,decoration={brace,amplitude=7pt,mirror,raise=5ex}]
		(-3,0.05) -- (-1.6, 0.05) node[midway,yshift=-3.5em]{segregation};
		\draw [decorate,decoration={brace,amplitude=7pt,mirror,raise=5ex}]
		(-1.6,0.05) -- (3, 0.05) node[midway,yshift=-3.5em]{pooling};
	\end{axis}
\end{tikzpicture}
\caption{Optimal Pack-Opponents-and-Pool Districting}
\label{f:sp}
\end{figure}
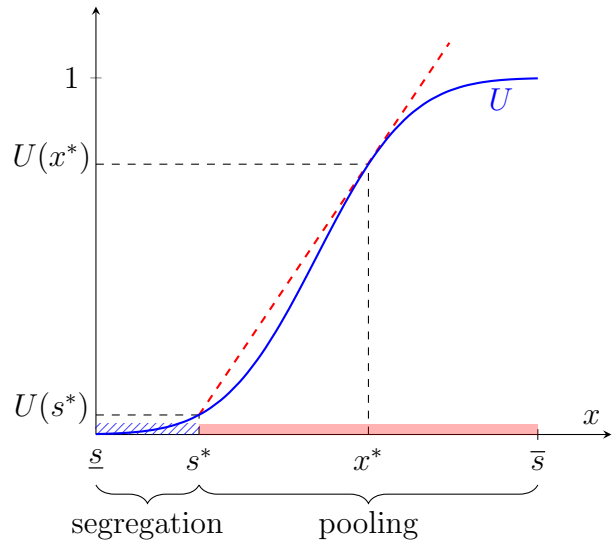

A simple result is that pack-opponents-and-pool is optimal when $U$ is strictly S-shaped.

\begin{proposition} \label{p:linearW} In the linear case where $U$ is strictly S-shaped, pack-opponents-and-pool districting is optimal, and every optimal districting plan has the same distribution of district means.
\end{proposition}
Intuitively, when $U$ is S-shaped the designer is risk-loving in the pool mean $x$ for $x\in [0,s^*]$ and is risk-averse in $x$ ``on average'' for $x\in [s^*,1]$, so voters below $s^*$ are segregated and voters above $s^*$ are pooled. Similar results were established by \citet{GP} and, in the persuasion literature, \citet{Kolotilin2017} and \citet{KMZ}.

As aggregate uncertainty vanishes, the best pack-opponents-and-pool plan converges to the plan characterized in Proposition \ref{nounc} with segregated packed districts.\footnote{Note that as $G$ converges to the step function $\mathbf{1}\{r\geq r^0\}$, $U$ converges to the step function $\mathbf{1}\{x\geq {x}^i\}$, where ${x}^i$ is the solution to $v({x}^i,r^0)=1/2$. The first-order condition then reduces to the condition that $x^*={x}^i$, which yields the same condition for $s^*$ as in Proposition \ref{nounc}.} Thus, traditional pack-and-crack (where packed districts are pooled) and pack-opponents-and-pool and POP (where packed districts are segregated) are all optimal without aggregate uncertainty, but only the latter two plans remain optimal with a small amount of aggregate uncertainty.\footnote{Intuitively, the designer optimally segregates packed districts to have a respectable chance of winning the strongest of these districts.} Note that pack-opponents-and-pool and POP induce the same distribution of district mean types, and hence may both be optimal even when the optimal distribution of means is unique. However, the designer's indifference among different ways of creating cracked districts with the same mean type is not robust to introducing slightly non-linear swings, as we show in the next section.

\begin{remark}[\textbf{Means vs. Medians}] An intuition for why packing opponents is optimal with linear swings and unimodal aggregate shocks (including in the no-aggregate-uncertainty case), while packing moderates is optimal with no idiosyncratic uncertainty, is that the designer targets a distribution of district means in the former case and district medians in the latter case. Optimizing the distribution of district means with unimodal aggregate uncertainty entails packing opponents and cracking moderates and supporters among districts with the same mean type. Optimizing the distribution of district medians entails matching voter types above and below the population median. Loosely speaking, whether packing opponents or moderates is optimal in practice depends on whether reality is closer to the linear/mean-dependent case or the no-idiosyncratic-uncertainty/median-dependent case.

The distinction between mean and median-dependence can be used to classify several strands of related literature. In gerrymandering, \citet{OG} and \citet{GP} study the mean-dependent case, while \citet{FH} study an approximately median-depedent case. In electoral competition, probabilisitic voting models with partisan taste shocks such as \citet{Hinich} and \citet{LindbeckWeibull} are mean-dependent, while stochastic median voter models such as \citet{Wittman} and \citet{Calvert} are median-dependent. In persuasion, \citet{GK-RS}, \citet{KMZL}, \citet{Kolotilin2017}, \citet{DM}, and \citet{KMS} study the mean-depedent case, while \citet{KCW} study a general case nesting both the mean and quantile (e.g., median)-dependent case, and \citet{YangZentefis} study the quantile-dependent case.

\end{remark}

\section{General Analysis} \label{s:nonlinear}

We now consider the general case with both idiosyncratic and aggregate uncertainty and non-linear swings. We first impose a natural curvature assumption on swings, and show that it implies that optimal districting is ``strictly single-dipped,'' in that more extreme voters are assigned to stronger districts. We then argue that optimal strictly single-dipped districting plans typically take a ``pack-and-pair'' form, where weaker districts are segregated and stronger districts consist of exactly two voter types. POP and PMP are leading examples of pack-and-pair plans. We next provide theoretical and numerical results that delineate the parameter ranges where POP or PMP is optimal. Here we find that POP is optimal when idiosyncratic uncertainty is much larger than aggregate uncertainty, PMP is optimal when aggregate uncertainty is larger than idiosyncratic uncertainty, and mixed versions of POP or PMP are optimal in the intermediate range. Finally, we observe that when idiosyncratic uncertainty is sufficiently dominant (as we will see is the case in practice), the optimal POP plan closely resembles $p$-segregation, and both $p$-segregation and traditional pack-and-crack districting are approximately optimal.

\subsection{Swingy Moderates and Single-Dipped Districting} \label{s:sdd}

The linear swing case considered in Section \ref{s:linear} is a natural benchmark, but it makes the counterfactual prediction that the ``swingiest'' voters---those with the largest $\Delta_{s}^{r,r'}$---are ``extremists'' with $s\in \{\ul s,\ol s\}$. In contrast, election forecasters (and presumably sophisticated gerrymanderers) take into account that moderate voters are usually swingier than extremists. As Nathaniel Rakich and Nate Silver put it when describing the ``elasticity scores'' in the FiveThirtyEight.com forecasting model, ``Voters at the extreme end of the spectrum---those who have a close to a 0 percent or a 100 percent chance of voting for one of the parties---don't swing as much as those in the middle,'' \citep{Silver}. We provide some evidence for this claim in Section \ref{s:estimation}.

The following assumption formalizes the idea that moderates are swingier than extremists.
\begin{assumption}[Swingy Moderates]
We have
\begin{equation} \label{swing}
\frac{\partial^2 }{\partial s\partial r} \ln \left(\frac{\partial v(s,r)}{\partial s}\right)>0 \quad \text{for all $s$, $r$}.
\end{equation}
\end{assumption}

To see why Assumption 1 corresponds to moderates being swingy, note that integrating \eqref{swing} gives, for all $s<s'<s''$ and $r<r'$,
\begin{align} 
&(v(s'',r')-v(s',r'))(v(s',r)-v(s,r))>(v(s'',r)-v(s',r))(v(s',r')-v(s,r')), \nonumber  
\end{align}
or equivalently
\begin{equation}
\Delta_{s'}^{r,r'} > \frac{v(s'',r)-v(s',r)}{v(s'',r)-v(s,r)}\Delta_{s}^{r,r'}+\frac{v(s',r)-v(s,r)}{v(s'',r)-v(s,r)}\Delta_{s''}^{r,r'} \quad \text{for all $s<s'<s''$, $r<r'$}. \label{swing2}	
\end{equation}
Recall that the linear case is defined by having equality in \eqref{swing2}. Thus, Assumption 1 says that, for any pair of aggregate shocks $r<r'$ and any triple of voter types $s<s'<s''$, when the aggregate shock improves from $r'$ to $r$, type $s'$ voters swing toward the designer more than type $s$ and $s''$ voters, relative to the linear case.

We mention an equivalent condition and an implication of Assumption 1.

\begin{proposition} \label{p:swingy} The following hold:
\begin{enumerate}
\item In the additive taste shock case, Assumption 1 holds iff the density $q$ of the taste shock $t$ is strictly log-concave: \[\frac{d^2 }{d t^2}\ln \left(q(t) \right)<0 \quad \text{for all $t$}.\]
\item Assumption 1 implies that $\partial v(s,r) / \partial r$ is strictly single-dipped (i.e., decreasing and then increasing) in $s$, for each $r$.
\end{enumerate}
\end{proposition}

Many common distributions have strictly log-concave densities, including the normal, logistic, and extreme value distributions (see, e.g., Table 1 in \citealt{BB}), so part 1 of the proposition shows that Assumption 1 is a standard property. The property in part 2 of the proposition gives another sense in which moderates are swingier than extremists. For example, for any $s<s'<s''$, this property implies that (letting $v_r = \partial v/\partial r$) if $v_r (s,r) =v_r (s'',r)$, then $v_r(s',r)<v_r (s,r) =v_r (s'',r) <0$ (recalling that $v_r <0$), so type $s'$ is swingier than types $s$ and $s''$.

We now show that Assumption 1 implies that every optimal districting plan is ``strictly single-dipped,'' in that more extreme voters are assigned to stronger districts. Formally, a districting plan $\mathcal{H}$ is \emph{strictly single-dipped} if any district $P\in \supp(\mathcal{H})$ containing any two voter types $s<s''$ is stronger than any district $P' \in \supp(\mathcal{H})$ containing any intervening voter type $s' \in (s,s'')$, in that $r^*(P')<r^*(P)$.\footnote{Formally, we say that a district $P$ ``contains'' a voter type $s$ if $s \in \supp(P)$.} Note that if districting is strictly single-dipped then each district consists of at most two distinct voter types.

\begin{proposition}\label{p:sdd}
Under Assumption 1, every optimal districting plan is strictly single-dipped.
\end{proposition}

Similar results were established by \citet{FH} and, in the persuasion context, \citet{KCW}.\footnote{Assumption 1 is equivalent to \citeauthor{FH}'s ``informative signal property.'' \citeauthor{FH} assume a finite number of districts, and also assume that the median and mode of $Q$ coincide. \citet{KCW} give sufficient conditions for single-dippedness in a more general model that allows state-dependent designer preferences.} To see the intuition, suppose a districting plan creates two districts, 1 and 2, with the same threshold aggregate shock $r^*$, but where District 1 consists entirely of moderates and District 2 consists of a mix of left-wing and right-wing extremists. With linear swings, the distribution of vote shares in the two districts are identical. However, under Assumption 1, the vote share is swingier in District 1 than in District 2. Thus, conditional on the aggregate shock being close to $r^*$, a marginal voter is more likely to be pivotal in District 2 than in District 1. The designer can then profitably exploit this asymmetry by re-allocating some unfavorable voters to District 1 and re-allocating some favorable voters to District 2, thus weakening the moderate District 1 and strengthening the extreme District 2. Breaking all ties in favor of extreme disticts in this manner leads to strictly single-dipped districting.

Proposition \ref{p:sdd} implies that, under Assumption 1, the designer should never pool more than two voter types in the same district. Thus, among the plans in Figure \ref{f:4pc}, only POP and PMP can be optimal under Assumption 1 (and, moreover, more extreme paired districts under these plans must be stronger than more moderate districts). In particular, while pack-opponents-and-pool is optimal with linear swings and unimodal aggregate shocks, if moderates are even slightly swingier than extremists then the designer is better-off splitting the pool into distinct districts each consisting of at most two types such that more extreme districts are strictly stronger.

\subsection{Pack-and-Pair Districting} \label{s:segpair}

Strict single-dippedness is an important property of a districting plan, but many plans can be strictly single-dipped. This subsection argues that, among strictly single-dipped plans, it is natural to focus on ``pack-and-pair'' districting, where weaker districts are segregated and stronger districts consist of exactly two voter types. Formally, a strictly single-dipped districting plan $\mathcal{H}$ is \emph{pack-and-pair} if $\delta_s \in \supp(\mathcal{H})$ implies that any $P \in \supp(\mathcal{H})$ such that $r^*(P)<r^*(\delta_s)$ takes the form $P=\delta_{s'}$ for some $s'<s$.

For simplicity, for the remainder of the current section, we restrict attention to the additive taste shock case, and assume that the taste shock density is strictly log-concave and symmetric about $0$. The symmetry assumption has the convenient implication that the threshold shock to win a packed district $P=\delta_s$ is just $r^*(P)=s$.

We first show that any pack-and-pair plan $\mathcal{H}$ can be described in a simple way. First, there exists a \emph{bifurcation point} $r^b\in [\ul s,\ol s]$ such that a district $P \in \supp(\mathcal{H})$ is packed if $r^*(P)\leq r^b$ and is paired if $r^*(P)> r^b$. The bifurcation point thus divides the packed and paired districts. Second, the assignment of voters to paired districts is described by a decreasing function $s_1$ and an increasing function $s_2$ where, for each paired district $P$, the two voter types in district $P$ are $s_1(r^*(P))$ and $s_2(r^*(P)) > s_1(r^*(P))$. Stronger paired districts thus contain more extreme voters, as single-dippedness requires.

\begin{proposition}\label{p:pp}
For any pack-and-pair districting plan $\mathcal H$, there exists a bifurcation point $r^b\in [\ul s,\ol s]$, a decreasing function $s_1: (r^b,\ol s]\rightarrow [\ul s,r^b)$, and an increasing function $s_2:(r^b,\ol s]\rightarrow (r^b,\ul s]$ satisfying $s_1(r)<r <s_2(r)$, such that for each $P\in \supp(\mathcal H)$, we have $\supp(P)=\{r^*(P)\}$ if $r^*(P)\leq r^b$ and $\supp(P)=\{s_1(r^*(P)),s_2(r^*(P))\}$ if $r^*(P)>r^b$. 
\end{proposition}

Examples of pack-and-pair districting include segregation, POP, PMP, and negative assortative districting. Note that segregation and negative assortative districting represent the extreme pack-and-pair plans where all voter types are segregated and where only a single type is segregated. We first give conditions under which these extreme districting plans are optimal. 

\begin{proposition} \label{p:segNAD}
Negative assortative districting is uniquely optimal if $G$ is concave, and segregation is uniquely optimal if $G$ is ``sufficiently convex,'' in that there exists a constant $c >0$ such that segregation is uniquely optimal if  ${g'(r)}/{g(r)}\geq c$ for all $r$.
\end{proposition}

The intuition for the first part of the result is as follows. First, any strictly single-dipped districting plan that never segregates any two voter types is negative assortative. So, it suffices to show that if $G$ is concave (and the taste shock density is strictly log-concave and symmetric), it is sub-optimal for the designer to segregate any two voter types $s<s'$. To see this, suppose the designer pools a few type-$s$ voters in with the type-$s'$ voters. The marginal effect of this change on the designer's expected seat share among type-$s$ voters is \[G(s')-G(s),\] which is the increased probability of winning a type-$s$ voter's district when she moves from the weak district $\delta_s$ to the strong district $\delta_{s'}$. On the other hand, the marginal effect of this change on the designer's expected seat share among type-$s'$ voters is
\[\frac{Q(s-s')-\frac 12}{q(0)}g(s').\]
This follows because the first term is the marginal effect on the threshold shock to win the strong district, where this comes from using the implicit function theorem (and $Q(0)=1/2$) to calculate $dr/d\rho$ at $\rho=0$ from the equation
\[\rho Q(s-r)+(1-\rho)Q(s'-r)=\frac{1}{2},\]
and the second term is the density of the aggregate shock at $r^*(\delta_{s'})=s'$. Finally, the sum of the two effects is positive, because
\[\frac{G(s')-G(s)}{g(s')} \geq s'-s > \frac{\frac{1}{2}- Q(s-s')}{q(0)},\]
where the first inequality is by concavity of $G$, and the second inequality is by symmetry and strict convexity of $Q$ on $(-\infty,0]$ (which follows from strict log-concavity of $q$).

The intuition for the second part of the result is that if $G$ is sufficiently convex then, for any two voter types $s$ and $s'$, we have 
\[
\frac{G(s')-G(s)}{g(s')}\leq \frac{Q(s'-s)-\frac{1}{2}}{q(0)},
\]
which by a similar logic as above implies that it is optimal for the designer to separate any two voter types rather than pooling them.

Proposition \ref{p:segNAD} expresses the intuition that concavity of $G$ favors pooling (which, under strict single-dippedness, takes the form of pairing types, rather than pooling intervals of types), while convexity of $G$ favors segregation. %
In the realistic case where $G$ is strictly S-shaped (i.e., the aggregate shock is unimodal), segregation and negative assortative districting are both sub-optimal, unless the two parties are substantially asymmetric.\footnote{Proposition \ref{p:nosegNAD} can be compared to Proposition 1 of \citet{FH}. \citeauthor{FH} show that PMP (``matching slices'') is optimal when idiosyncratic uncertainty is sufficiently small, but their discussion focuses on the extreme case of negative assortative districting, where only a single voter type is segregated. Proposition \ref{p:nosegNAD} shows that this extreme case never arises when the distribution of the aggregate shock is unimodal and the two parties are symmetric.}

\begin{proposition} \label{p:nosegNAD}
If $G$ is strictly S-shaped with inflection point $r^*(F)$, then segregation and negative assortative districting are both sub-optimal.
\end{proposition}
The intuition is simple. By Proposition \ref{p:segNAD}, the designer prefers pooling any two voter types above the inflection point $r^*(F)$, so segregation is suboptimal. Moreover, for any negative assortative districting, there exist nearby voter types that are paired in a district $P$ with $r^*(P)<r^*(F)$, but the designer prefers segregating such types.

Since convexity of $G$ favors segregation, concavity of $G$ favors pairing, and it is natural to assume that $G$ is S-shaped (first convex, then concave), a natural conjecture is that pack-and-pair districting (first segregation, then pairing) is optimal. We can verify this conjecture numerically (for an extremely wide range of parameters) in the special case where $G$ and $Q$ are both normal. The following proposition states this result, as well as giving a general sufficient condition for pack-and-pair districting to be uniquely optimal.

\begin{proposition}\label{p:pap}
If there do not exist $\ul s\leq s<r<s'\leq s''\leq \ol s$ satisfying
\begin{equation}\label{e:pap}
\begin{gathered}
G(r)+\lambda(r)\left (Q(s-r)-\tfrac 12 \right )\geq G(s)\quad \text{and}\\
G(r)+\lambda(r)\left (Q(s-r)-\tfrac 12 \right )\geq G(s'')+\lambda(s'')\left(Q(s-s'')-\tfrac 12 \right),	
\end{gathered}
\end{equation}
where
\[
\lambda(r)= \frac{g(r)(Q(s'-r))-Q(s-r))}{\left(Q(s'-r)-\frac 12\right)q(s-r)-\left(Q(s-r)-\frac 12\right)q(s'-r)} \quad \text{and}\quad \lambda(s'')=\frac{g(s'')}{q(0)},
\]

then every optimal districting plan is pack-and-pair. Moreover, when $Q$ is the standard normal distribution and $G$ is the centered normal distribution with standard deviation $\gamma^{-1}$, there do not exist $\gamma\in \{.1,.2,\ldots, 99.9,100\}$ and $s<r<s'\leq s''$ with $s,r,s',s''\in \{-5,-4.9,\ldots,4.9,5\}$ that satisfy \eqref{e:pap}.
\end{proposition}

Condition \eqref{e:pap} can be explained as follows. For any (strictly single-dipped) non-pack-and-pair plan, there exist $s<r<s'\leq s''$ such that voter types $s<s'$ are paired in a district $P$ with $r^* (P)=r\in (s,s')$ and voter type $s''$ is segregated. By a similar logic to Proposition \ref{p:segNAD}, if the first inequality in \eqref{e:pap} fails, the designer prefers to segregate a few type-$s$ voters from district $P$; and if the second inequality in \eqref{e:pap} fails, the designer prefers to move a few type-$s$ voters from district $P$ to district $\delta_{s''}$.
Thus, if there do not exist $s<r<s'\leq s''$ that satisfy \eqref{e:pap}, then any optimal plan must be pack-and-pair.

\subsection{Should Opponents or Moderates be Packed?}\label{s:Y}

Having provided some arguments for pack-and-pair districting, the last part of our analysis compares two key forms of pack-and-pair---POP and PMP---as well as mixed versions of these districting plans. The mixed versions of POP and PMP that we will encounter fall into a class of districting plans that we call ``Y-districting.'' Formally, a pack-and-pair plan $\mathcal{H}$ is \emph{Y-districting} if there exists a positive number $\varepsilon >0$ such that
\begin{enumerate}
\item For all $r \in [r^b -\varepsilon,r^b +\varepsilon]$ (where $r^b$ is the bifurcation point), there exists $P \in \supp (\mathcal{H})$ such that $r^*(P)=r$.
\item Districts $P\in \supp(\mathcal H)$ with $r^*(P)\in [r^b -\varepsilon, r^b]$ are segregated (i.e., $\supp(P)=\{r^*(P)\}$).
\item Districts $P\in \supp(\mathcal H)$ with $r^*(P)\in (r^b, r^b +\varepsilon]$ are paired (i.e., $\supp(P)=\{s_1 (r^*(P)),s_2 (r^*(P))\}$ for some $s_1 (r^*(P))<s_2 (r^*(P))$).
\item The functions $s_1$ and $s_2$ describing the voter types in paired districts are twice differentiable and satisfy $\lim_{r\downarrow r^b}s_1(r)=\lim_{r\downarrow r^b}s_2(r)$.\footnote{The differentiability condition is used in the proof of Proposition \ref{p:gamma}. It may be possible to drop~it.}
\end{enumerate}
We will see that Y-districting encompasses a mixed version of POP, where there exists $\hat s \in (\ul s, r^b)$ such that voter types in $[\ul s, \hat s)$ are always segregated and types in $(\hat s, r^b)$ are sometimes segregated and sometimes paired, as well as a mixed version of PMP, where there exists $\hat s \in (\ul s, r^b)$ such that types in $[\ul s, \hat s)$ are always paired and types in $(\hat s,r^b)$ are sometimes segregated and sometimes paired. (In contrast, recall that under POP there exists $\hat s \in (\ul s, r^b)$ such that types in $[\ul s, \hat s)$ are always segregated and types in $(\hat s, r^b)$ are always paired, while under PMP there exists $\hat s \in (\ul s, r^b)$ such that types in $[\ul s, \hat s)$ are always paired and types in $(\hat s, r^b)$ are always segregated.) We will give theoretical and numerical results that indicate that POP is optimal when idiosyncratic uncertainty is much larger than aggregate uncertainty, PMP is optimal when aggregate uncertainty is larger than idiosyncratic uncertainty, and Y-districting (and, in particular, mixed POP or mixed PMP) is optimal in the intermediate range.

We first discuss how POP, PMP, and Y-districting relate to the set of all pack-and-pair plans. POP and PMP are both \emph{pure} districting plans, in that each voter type $s$ is assigned to a single district $P$: formally, for each $s \in [\ul s,\ol s]$, there exists a unique $P \in \supp (\mathcal{H})$ such that $s \in \supp (P)$. They are not the only pure districting plans: for example, a pack-and-pair plan could segregate voter types below a cutoff $s_0$ and match slices (including with an intermediate segregation region) among voter types above $s_0$. However, POP and PMP are the simplest such plans, as they involve only a single non-degenerate interval of segregated voter types. %
We are not aware of any parameters for which a more complex pure pack-and-pair plan is optimal. %

In contrast, Y-districting plans are mixed, because voter types $s$ just below the bifurcation point are sometimes segregated and sometimes paired with higher types. Somewhat surprisingly, we will see that such plans are uniquely optimal for a range of parameters, even though voter types are continuous. While not every mixed pack-and-pair plan is Y-districting, we will see that, at least numerically, optimal plans always take one of the three forms we consider.

We would like to have general necessary and sufficient conditions for the optimality of POP, PMP, and Y-districting. Unfortunately, this seems very challenging, because the form of optimal districting is driven by global constraints that are difficult to analyze. We instead present a seemingly modest result, which is that if Y-districting is optimal, then the ratio of idiosyncratic uncertainty to aggregate uncertainty must fall in an intermediate range. However, numerically it appears that this result actually characterizes when all three forms of districting are optimal: at least in the case where aggregate and idiosyncratic shocks are both normally distributed, our necessary conditions for optimality of Y-districting are also approximately sufficient, and when the ratio of idiosyncratic uncertainty to aggregate uncertainty is below (resp., above) the range where Y-districting is optimal, then PMP (resp., POP) is optimal.

To facilitate a comparison of the amount of aggregate and idiosyncratic uncertainty, the distributions $G$ and $Q$ should have the same shape. We therefore assume that there exists a parameter $\gamma >0$ such that $G(r)=Q(\gamma r)$ for all $r$. The parameter $\gamma$ thus meaures the ratio of the standard deviation of the idiosyncratic shocks (which is normalized to $1$) to that of the aggregate shock (which equals $\gamma^{-1}$). The following is our key result.

\begin{proposition} \label{p:gamma}
If Y-districting is optimal, then $r^b=0$ and $\gamma \in (1,\sqrt{1+\sqrt{3}}\approx 1.65]$.
\end{proposition}

The proof of Proposition \ref{p:gamma} proceeds by deriving three necessary conditions for optimal districting to involve a bifurcation point at $r$ (which are based on linear programming duality), and then showing that these conditions imply that the bifurcation point must coincide with the inflection point, and the ratio of idiosyncratic to aggregate uncertainty must lie in an intermediate range. The first condition (equation \eqref{e:Y1} in Appendix \ref{a:B}) says that it is optimal to pair voter types just below and just above $r$. The second condition (equation \eqref{e:Y2}) says that it is optimal to segregate types just below $r$. The third condition (equation \eqref{e:Y3}) says that the proportions of favorable and unfavorable voters in each district $P$ with $r^*(P)=r'$ just above $r$ actually generate the desired cutoff $r'$. Intuitively, for it to be optimal to pair nearby voter types around $r$, $G$ must be weakly concave at $r$; and for it to be optimal to segregate voter types just below $r$, $G$ must be weakly convex at $r$. Hence, bifurcation can occur only at the inflection point of $G$, which by symmetry equals $0$. Moreover, if we take parameters where Y-districting is optimal and increase aggregate uncertainty, it eventually becomes optimal to always segregate voter types just below $0$ rather than pairing them with higher voter types, at which point optimal districting becomes PMP (with a bifurcation point below $0$). On the other hand, if we take parameters where Y-districting is optimal and decrease aggregate uncertainty, it eventually becomes optimal to always pair voter types just below $0$ with higher voter types rather than segregating them, at which point optimal districting becomes POP (with a bifurcation point above $0$). We discuss the mechanics of the transition from PMP to POP as $\gamma$ increases below.

If we take for granted that the condition $\gamma \in (1,1.65)$ is sufficient as well as necessary for Y-districting to be optimal, the above intuition suggests that:
\begin{enumerate}
\item PMP is optimal when $\gamma \leq 1$.
\item Y-districting is optimal when $\gamma \in (1,1.65)$.
\item POP is optimal when $\gamma \geq 1.65$.
\end{enumerate}
Figure \ref{fig:numeric} presents numerical solutions that verify this heuristic. In the figure, $Q$ is the standard normal distribution, $G$ is the centered normal distribution with standard deviation $\gamma^{-1}$, and $F$ is the uniform distribution on $[-1,1]$.\footnote{More precisely, we approximate the designer's problem by a finite-dimensional linear program and then solve it using Gurobi Optimizer. Our approximation specifies that $s$ is uniformly distributed on $\{-1,-.99,\ldots,.99, 1\}$ and that the designer is constrained to create districts $P$ satisfying $r^*(P)\in \{-1,-.99,\ldots,.99, 1\}$.} Voter types are on the $x$-axis, and the threshold shocks to win the districts to which each voter type is assigned are on the $y$-axis. (Thus, packed districts lie on the $45^\circ$ line, while paired districts straddle the $45^\circ$ line.) For mixed districting plans (i.e., Y-districting, the middle row of the figure), the shading intensity indicates the probability that a voter type is assigned to each district. We see that optimal districting takes exactly the conjectured form: PMP is optimal for $\gamma \in \{0.2,0.5,1\}$, Y-districting is optimal for $\gamma \in \{1.2,1.4,1.6\}$, and POP is optimal for $\gamma \in \{1.7,3,6\}$. The highest value of $\gamma$ in the figure, $\gamma=6$, is the value closest to our empirical estimates. When $\gamma =6$, POP remains optimal but now closely resembles $p$-segregation. Thus, for what we will see is the empirically relevant parameter range, $p$-segregation is approximately optimal. %

\begin{figure}
    \centering
    \begin{subfigure}{0.31\textwidth}
        \includegraphics[width=\linewidth]{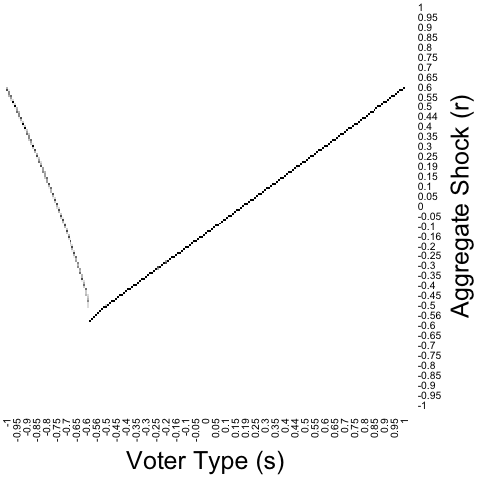}
        \caption{$\gamma=0.2$}
        \label{fig:panel1}
    \end{subfigure}
    \hspace{0.018\linewidth}
    \begin{subfigure}{0.31\textwidth}
        \includegraphics[width=\linewidth]{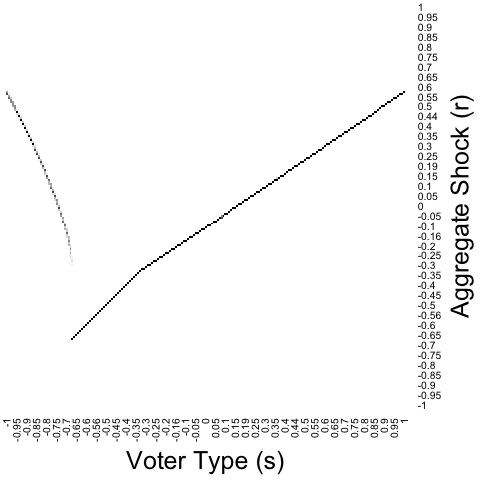}
        \caption{$\gamma=0.5$}
        \label{fig:panel2}
    \end{subfigure}
    \hspace{0.018\linewidth}
    \begin{subfigure}{0.31\textwidth}
        \includegraphics[width=\linewidth]{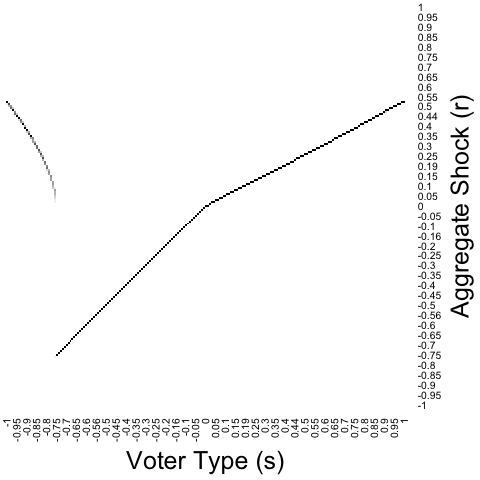}
        \caption{$\gamma=1$}
        \label{fig:panel3}
    \end{subfigure}
    
    \bigskip
    
    \begin{subfigure}{0.31\textwidth}
        \includegraphics[width=\linewidth]{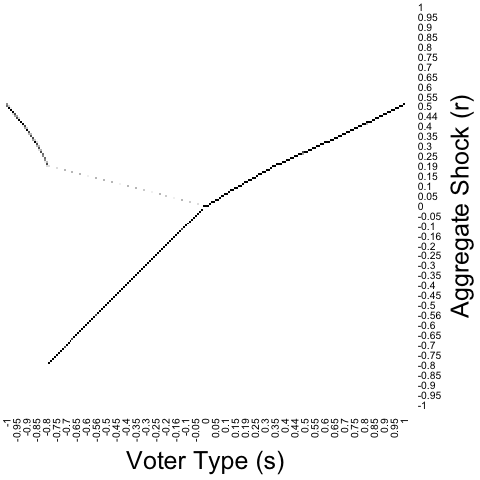}
        \caption{$\gamma=1.2$}
        \label{fig:panel4}
    \end{subfigure}
    \hspace{0.018\linewidth}
    \begin{subfigure}{0.31\textwidth}
        \includegraphics[width=\linewidth]{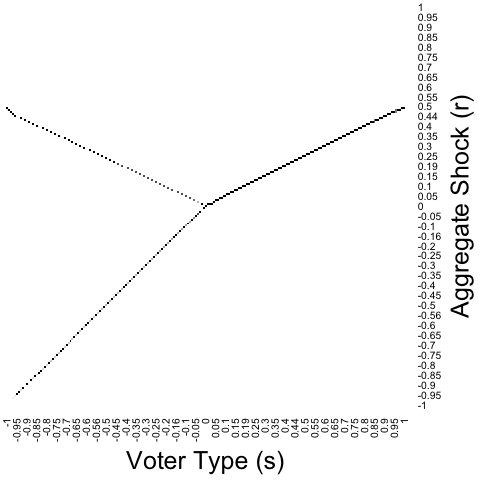}
        \caption{$\gamma=1.4$}
        \label{fig:panel5}
    \end{subfigure}
    \hspace{0.018\linewidth}
    \begin{subfigure}{0.31\textwidth}
        \includegraphics[width=\linewidth]{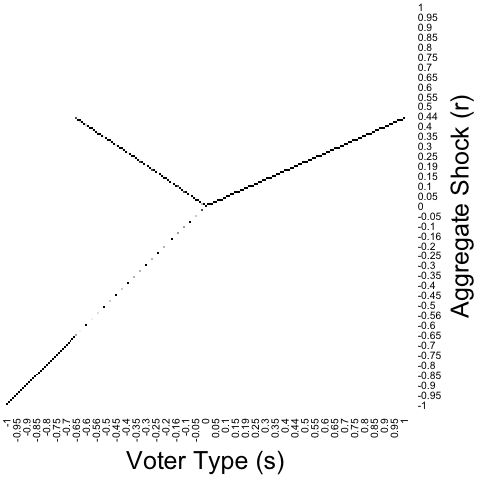}
        \caption{$\gamma=1.6$}
        \label{fig:panel6}
    \end{subfigure}
    
    \bigskip
    
    \begin{subfigure}{0.31\textwidth}
        \includegraphics[width=\linewidth]{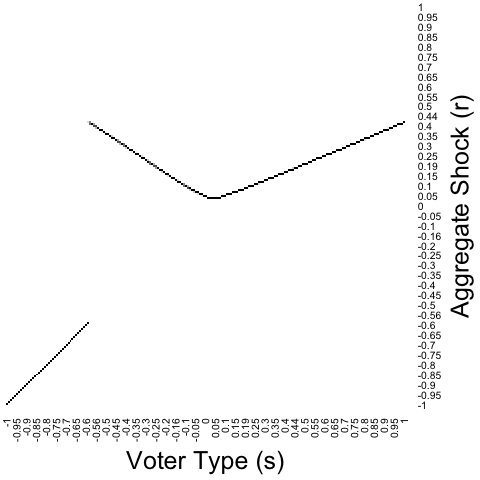}
        \caption{$\gamma=1.7$}
        \label{fig:panel7}
    \end{subfigure}
    \hspace{0.018\linewidth}
    \begin{subfigure}{0.31\textwidth}
        \includegraphics[width=\linewidth]{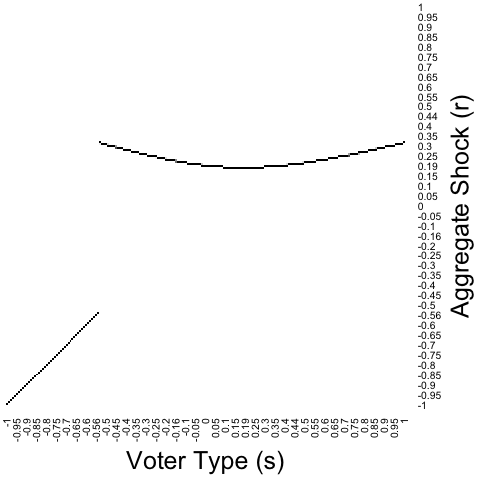}
        \caption{$\gamma=3$}
        \label{fig:panel8}
    \end{subfigure}
    \hspace{0.018\linewidth}
    \begin{subfigure}{0.31\textwidth}
        \includegraphics[width=\linewidth]{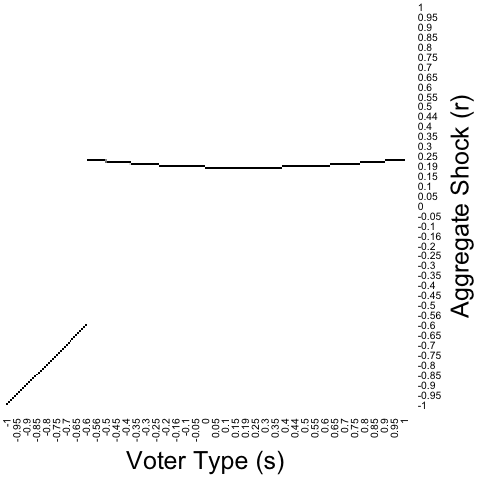}
        \caption{$\gamma=6$}
        \label{fig:panel9}
    \end{subfigure}
    
    \caption{Optimal Districting as $\gamma$ Varies}
    \caption*{\emph{Notes:} The optimal districting plan is PMP for $\gamma \in \{0.2,0.5,1\}$, Y-districting for $\gamma \in \{1.2,1.4,1.6\}$ (and, specifically, mixed PMP for $\gamma \in \{1.2,1.4\}$ and mixed POP for $\gamma = 1.6$), and POP for $\gamma \in \{1.7,3,6\}$. Our empirical estimates of $\gamma$ in Section \ref{s:estimation} are above $6$.}
    \label{fig:numeric}
\end{figure}

We can give an intuition for how and why optimal districting transitions from PMP to POP as $\gamma$ increases, as illustrated in Figure \ref{fig:numeric}. Along the way, we also mention some additional features of optimal PMP and POP plans, as well as describing the transition from mixed PMP to mixed POP within the Y-districting regime.

First, recall the extreme cases where $\gamma$ is close to $0$ (almost no idiosyncratic uncertainty) and where $\gamma$ is very large (almost no aggregate uncertainty). When $\gamma$ is close to $0$, PMP is optimal; moreover, when $F$ is symmetric about $0$ as in Figure \ref{fig:numeric}, almost all voters are paired, so optimal districting is approximately negative assortative, which implies that the bifurcation point is below $0$ and the range of values of $r^*(P)$ across paired districts $P$ is large.\footnote{Another property of optimal PMP plans is that the left arm of the ``Y'' is infinitely steep at the bifurcation point, i.e., $\lim_{r\downarrow r^b}s_1'(r)=0$.} When $\gamma$ is very large, POP is optimal; moreover, $p$-segregation is approximately optimal, which implies that the bifurcation point is above $0$ and the range of values of $r^*(P)$ across paired districts is very small.\footnote{Another property of optimal POP plans is that pairing at the bifurcation point is smooth, i.e., $\lim_{r\downarrow r^b}s_1'(r)=-\infty$ and $\lim_{r\downarrow r^b}s_2'(r)=\infty$.} Now, when $\gamma$ increases from $0$ toward $1$, the range of $r^*(P)$ across paired districts decreases (as the range of probable aggregate shocks decreases), and the proportion of packed districts increases. When $\gamma$ reaches $1$, it becomes optimal to pack voters with $s=0$, the inflection point of $G$. Since it cannot be optimal to pack voters above the inflection point, once $\gamma$ crosses $1$ it becomes optimal to pair voters with $s$ just above $0$ with a few slightly less favorable voters. At this point, districting takes the form of mixed PMP.

As $\gamma$ increases farther above $1$, the range of $r^*(P)$ across paired districts continues to decrease. This implies a flattening out of the right arm of the ``Y''---i.e., an increase in $s_2'$---which increases the mass of favorable voters assigned to districts where $r^*(P)$ is positive but small. To keep $r^*(P)$ small in these districts, this effect must be offset by also assigning more unfavorable voters to these districts, which is achieved by assigning more of the ``mixed'' unfavorable voters type to paired districts rather than packed districts, while the range of unfavorable voter types assigned to each interval of mixed districts actually decreases---i.e., the left arm of the Y gets steeper.\footnote{The proof of Proposition \ref{p:gamma} shows that, for all sufficiently small positive $r$, $|s_1'(r)|$ is decreasing in $\gamma$ (i.e., the left arm gets steeper) and $s_2'(r)$ is increasing in $\gamma$ (i.e., the right arm gets flatter). \label{f:gamma}} At some point, the right arm of the Y becomes flatter than the left arm so that the most extreme left-wing voters have no right-wing voters to match with, at which point these voters are segregated: this point marks the transition from mixed PMP to mixed POP, which occurs at $\gamma=\sqrt{2} \approx 1.41$ in the uniform case illustrated in Figure \ref{fig:numeric}.\footnote{The transition point $\gamma=\sqrt{2}$ is defined as the unique value of $\gamma$ at which $\lim_{r\downarrow 0} |s_1'(r)|=\lim_{r\downarrow 0} s_2'(r)$. The $\gamma=1.4$ panel in the figure illustrates a point just before this transition occurs.} As $\gamma$ increases further, more and more mixed unfavorable voters are assigned to paired districts, until all such voters are assigned to paired districts, at which point optimal districting becomes POP, and the bifurcation point becomes positive. This occurs when $\gamma \approx 1.65$. Finally, as $\gamma$ increases further beyond $1.65$, the range of $r^*(P)$ across paired districts continues to decrease, and the optimal POP plan approximates $p$-segregation more and more closely.

\begin{remark}[\textbf{Approximate Optimality of Traditional Pack-and-Crack}]
We conclude this setting by noting that, for what we will see is the empirically-relevant range of parameters, the optimal POP plan closely resembles $p$-segregation, and in fact both $p$-segregation and traditional pack-and-crack districting are approximately optimal. Our central estimates for $\gamma$ in Section \ref{s:estimation} are above 6, and for most states are above 10. Figure \ref{fig:numeric} shows that, for these parameters, POP is optimal, and the optimal POP plan closely resembles $p$-segregation. Moreover, for the parameters used in Figure \ref{fig:numeric} (where the standard deviation of $s$ is fixed at what we will see is a realistic level, while $\gamma^{-1}$, the standard deviation of $r$, varies), we have calculated that the designer's expected seat share under the optimal districting plan never exceeds his expected seat share under the optimal traditional pack-and-crack plan by more than $1.4\%$ for any value of $\gamma$, or by more than $0.1\%$ for any value of $\gamma$ above $5$.\footnote{\citeauthor{FH} (\citeyear*{FH}, p. 129) and  \citeauthor{CH} (\citeyear*{CH} p. 571) present an example with large aggregate uncertainty ($\gamma=1/\sqrt{2}\approx 0.71$) \emph{and} a large standard deviation of $s$ (equal to $3$, while our empirical estimate of this parameter is $0.63$) where the designer's expected seat share is over $20\%$ greater under matching slices than under traditional pack-and-crack. This shows that, when the standard deviations of \emph{both} $r$ and $s$ are (unrealistically) large, the advantage of optimal districting over traditional pack-and-crack can be significantly larger than the $1.4\%$ upper bound that we obtain by varying the standard deviation of $r$ while fixing the standard deviation of $s$ at a realistic level.} For example, when $\gamma=6$ the optimal expected seat share is approximately $.7087$, while the optimal traditional pack-and-crack plan gives an expected seat share of approximately $.7082$.\footnote{When $\gamma=2$ (an unrealistic low value), the corresponding expected seat shares are $.5392$ and $.5357$. When $\gamma=15$ (close to our central estimate), they are $.8488$ and $.8485$.} An intuition for this result is that in practice aggregate uncertainty is small (relative to both idiosyncratic uncertainty and the range of voter/precinct types $s$), so the no-aggregate uncertainty case considered in Section \ref{s:noaggr}---where traditional pack-and-crack is exactly optimal---is fairly realistic.
\end{remark}

\section{Estimation}\label{s:estimation}

We have argued that the form of optimal districting depends on a comparison of the amount of aggregate and idiosyncratic uncertainty facing the designer, and in particular on the parameter $\gamma$ introduced in the previous section (i.e., the ratio of idiosyncratic to aggregate uncertainty, or equivalently the inverse standard deviation of the aggregate shock $r$, recalling that the the standard deviation of the idiosyncratic shocks $t$ is normalized to $1$). We now estimate $\gamma$ using precinct-level returns from recent US House elections, while also providing empirical support for some of our key theoretical assumptions. We first describe our data and empirical model, then present some simple summary statistics and plots, and finally estimate $\gamma$.

\subsection{Data and Empirical Model} Our data are the precinct-level returns for the US House elections in 2016, 2018, and 2020, which were recently standardized and made freely available by \citet{Baltz}. For each precinct $n$ and election $t \in \{2016,2018,2020\}$, we observe the total two-party vote $k_{nt}$ and the share of the two-party vote for the Republican candidate $v_{nt}$.\footnote{A ``precinct'' is the smallest election-reporting unit in a state, which typically corresponds to a geographic area where all voters vote at the same polling place. Maine and New Jersey report election returns only at the township level, so for these states $n$ indexes townships rather than precincts. Also, for some elections where a nominally third-party candidate runs in place of an official Democratic or Republican candidate, we manually re-label this candidate as a Democrat or Republican. For example, in New York, we re-assign Working Families Party candidates as Democrats and re-assign Conservative Party candidates as Republicans.} The data are a repeated cross-section rather than a panel, because there is no general way to match precincts across elections (for example, because precinct boundaries change frequently; \citealt{Baltz}, p. 6). We drop all districts with an uncontested House race in any of 2016, 2018, or 2020 (which drops 25\% of all districts).\footnote{Keeping these districts would bias our estimate of $\gamma$, because the relevant vote shares are for contested elections, and if these districts were contested their vote shares would be different from 0 or 1. Keeping a district with one or two uncontested elections only for the elections where it is contested would also bias our estimate of $\gamma$, by distorting the estimated swing across elections. Dropping uncontested districts does likely bias our estimate of the distribution $F$ of voter types $s$, as uncontested districts are presumably more extreme; however, this bias is irrelevant for our main goal of estimating $\gamma$.} Moreover, for each of the three elections, we drop precincts where there are fewer than 50 total votes (which drops .13\% of all votes) or where the Republican vote share is 0 or 1 (which drops an additional .015\% of votes). 

To take the model to these data, we assume that the designer has voter information at the precinct level. This is a reasonable assumption, since this is the finest level at which election data is available. As a voter type $s$ in the model captures the information available to the designer, we therefore assume that all voters in a given precinct $n$ have the same type $s_n$. We will also assume that precincts are relatively large (in the data, the mean precinct vote count is 789 with standard deviation 1,399, after dropping precincts with fewer than 50 total votes or a 0 or 1 vote share), and idiosyncratic taste shocks are normally distributed, so that the designer's vote share in precinct $n$ in election $t$ is given by \[v(s_n,r_t)=\Phi \left( {s_n - r_t} \right),\] where $\Phi$ is the standard normal cdf. %

While our estimation relies on the assumption that taste shocks are normally distributed, it is important to note that our estimates are quite insensitive to this assumption: because we will find that $\gamma$ is very large, the taste shock distribution is approximately uniform over the relevant range, so specifying any smooth taste shock distribution leaves our estimates almost unchanged.

\subsection{Descriptive Figures and Summary Statistics} We first present a histogram (Figure \ref{fig:HistogramPrecinct}(a)) showing the number of voters in the United States who live in a precinct with Republican vote share $v$, with bin breaks  $\{0,.05,\ldots,.95,1\}$, averaging over elections $t \in \{2016,2018,2020\}$. The histogram shows that the distribution of $v_{nt}$ is unimodal, with a large majority (74\%) of the mass on $v \in [.25,.75]$. This pattern has two simple, but important, implications for our model. First, the distribution of voter/precinct types is far from bimodal: there is a continuum of types, with most mass ``toward the middle.'' A designer choosing how to partition precincts into districts must thus decide how to allocate a continuum of types, as in our model.\footnote{In practice, the smallest ``districtable unit'' is not a precinct but a \emph{census block}, which is the smallest geographic unit for which the US Census tabulates complete data. However, the number of voters in a precinct or a census block are roughly similar (typically around 1,000, albeit with fairly wide variation), so we believe there is little loss in proceeding as if designers partition precincts rather than census blocks.} Second, idiosyncratic uncertainty appears to be large relative to aggregate uncertainty. To see this, note that if idiosyncratic uncertainty were extremely large, Figure \ref{fig:HistogramPrecinct}(a) would show a degenerate distribution at $v=1/2$, while if aggregate uncertainty were extremely large, it would show a bimodal distribution with all mass at $0$ and $1$. The former case is a better approximation, as the actual distribution in Figure \ref{fig:HistogramPrecinct}(a) is unimodal, with 74\% of the mass on $v \in [.25,.75]$. While we will quantitatively estimate $\gamma$ in the next subsection, this observation already suggests what we will find, which is that $\gamma$ is much greater than $1$.

\begin{figure}[t]
  \centering
  \begin{subfigure}[b]{0.49\textwidth}
    \includegraphics[width=\textwidth]{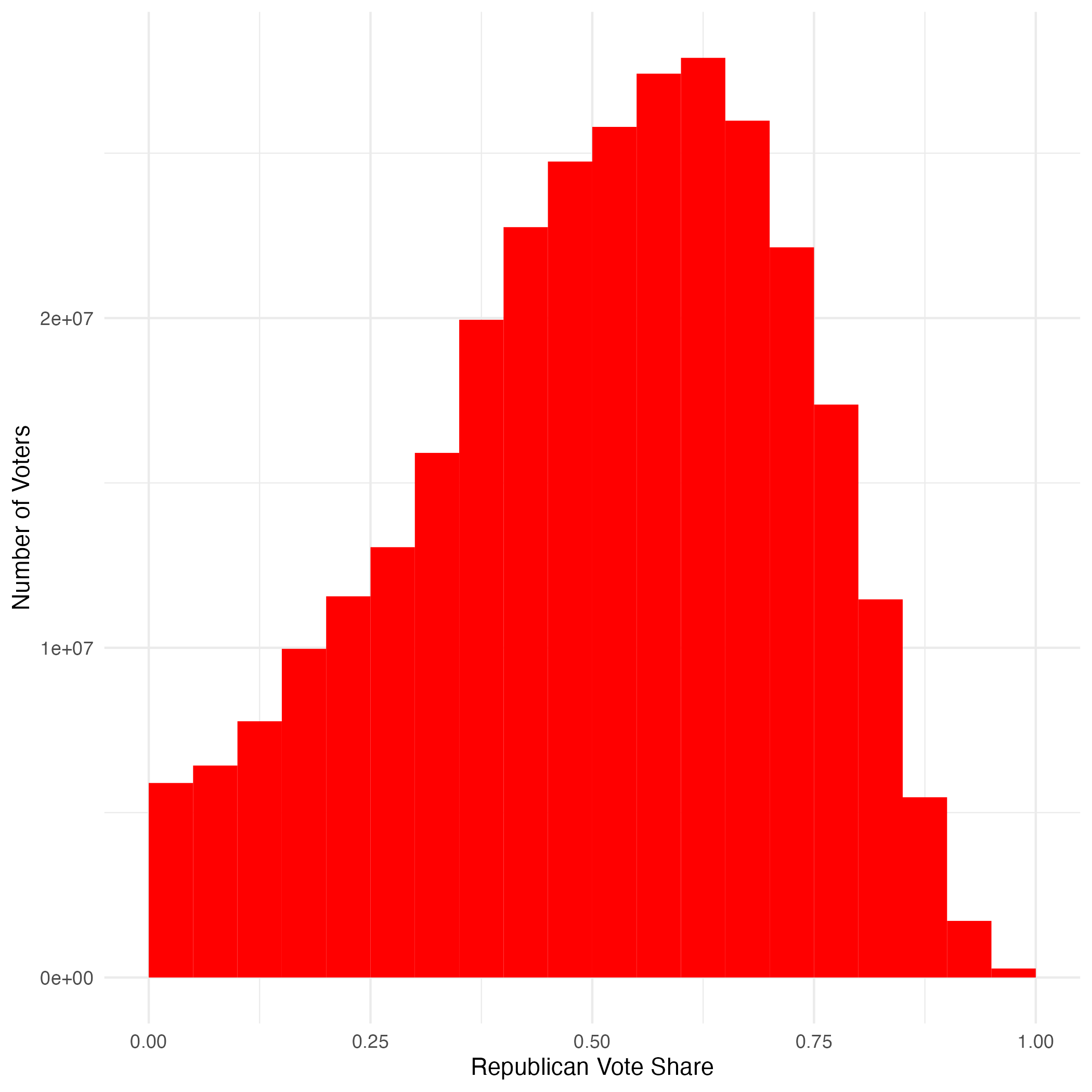}
    \captionsetup{labelformat=empty}
    \caption{(a) Precinct Vote Shares}
    \label{fig:HistogramPrecinct}
  \end{subfigure}
  \hfill
  \begin{subfigure}[b]{0.49\textwidth}
    \includegraphics[width=\textwidth]{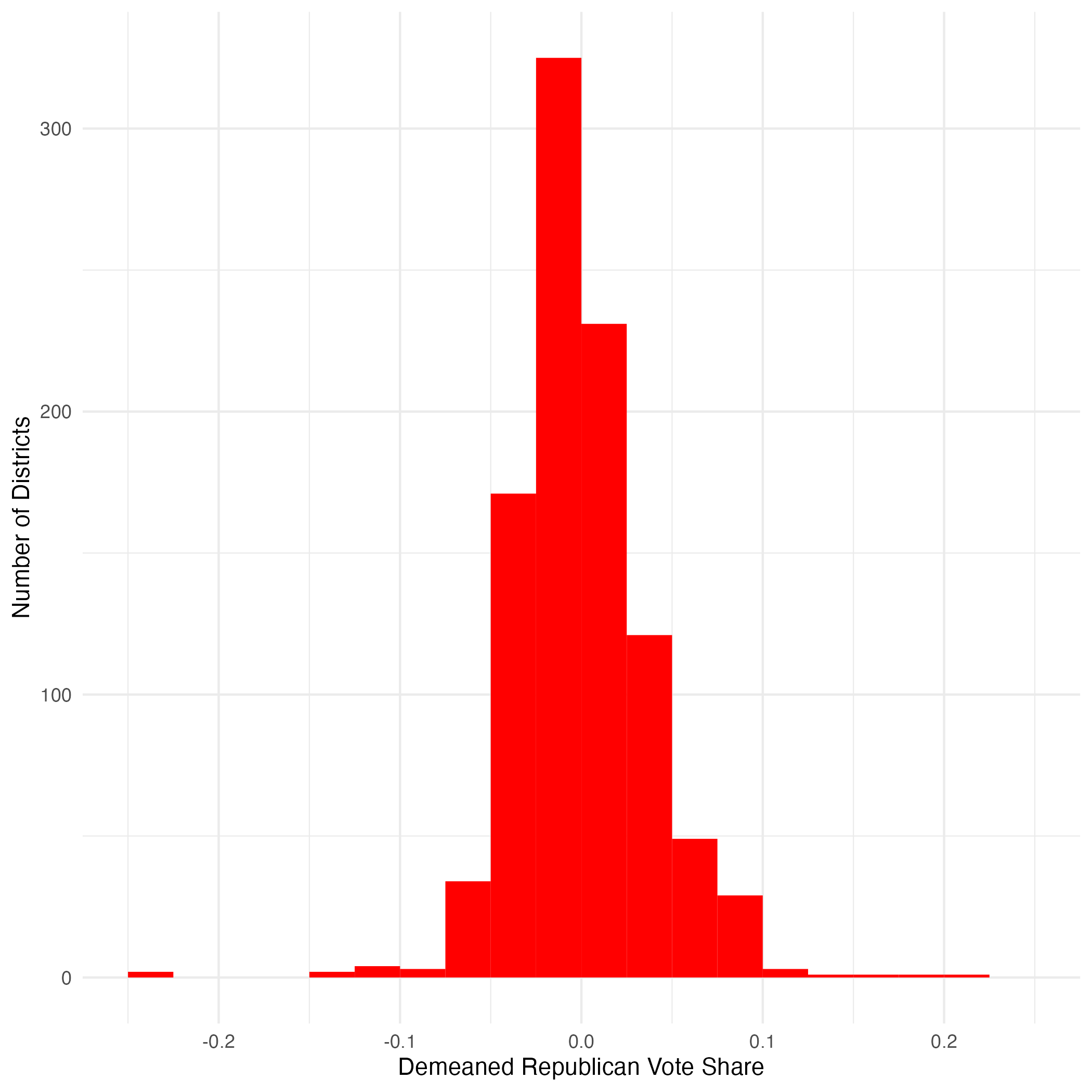}
    \captionsetup{labelformat=empty}
    \caption{(b) District Vote Swings}
    \label{fig:HistogramDistrict}
  \end{subfigure}
  \caption{Distributions of Precinct Vote Shares and District Vote Swings}
  \label{fig:Distributions}
\end{figure}

Next we present another histogram (Figure \ref{fig:HistogramDistrict}(b)), which shows the number of \emph{(district, election)} pairs where the district-wide Republican vote share deviated from its mean over the three elections we consider by $x$, with bin breaks $\{-.25,-.225,\ldots,.225,.25\}$.\footnote{This histogram is compiled at the district level because precincts are not matched across elections.} This histogram gives another way of showing that aggregate shocks are small: the distribution is centrally unimodal, and most of the mass (57\%) is on $x \in [-.025,.025]$. In contrast, if aggregate shocks were large, we would again have a bimodal distribution with all mass far from $0$.

Finally, we consider the empirical distribution of vote shares $v_{nt}$ across precincts $n$ (weighted by the number of votes in each precinct), for each election $t$. This is shown in Figure \ref{fig:CDFVoteShare}(a). The S-shaped curve for each election again indicates that most precincts have vote shares relatively close to $1/2$. The ordering of the curves (except for the lowest-vote-share precincts, discussed below) reflects the fact that, among the 2016, 2018, and 2020 elections, 2018 was the best year for Democrats, 2016 was the best year for Republicans, and 2020 was in the middle.

\begin{figure}[t]
  \centering
  \begin{subfigure}[b]{0.49\textwidth}
  \includegraphics[width=\textwidth]{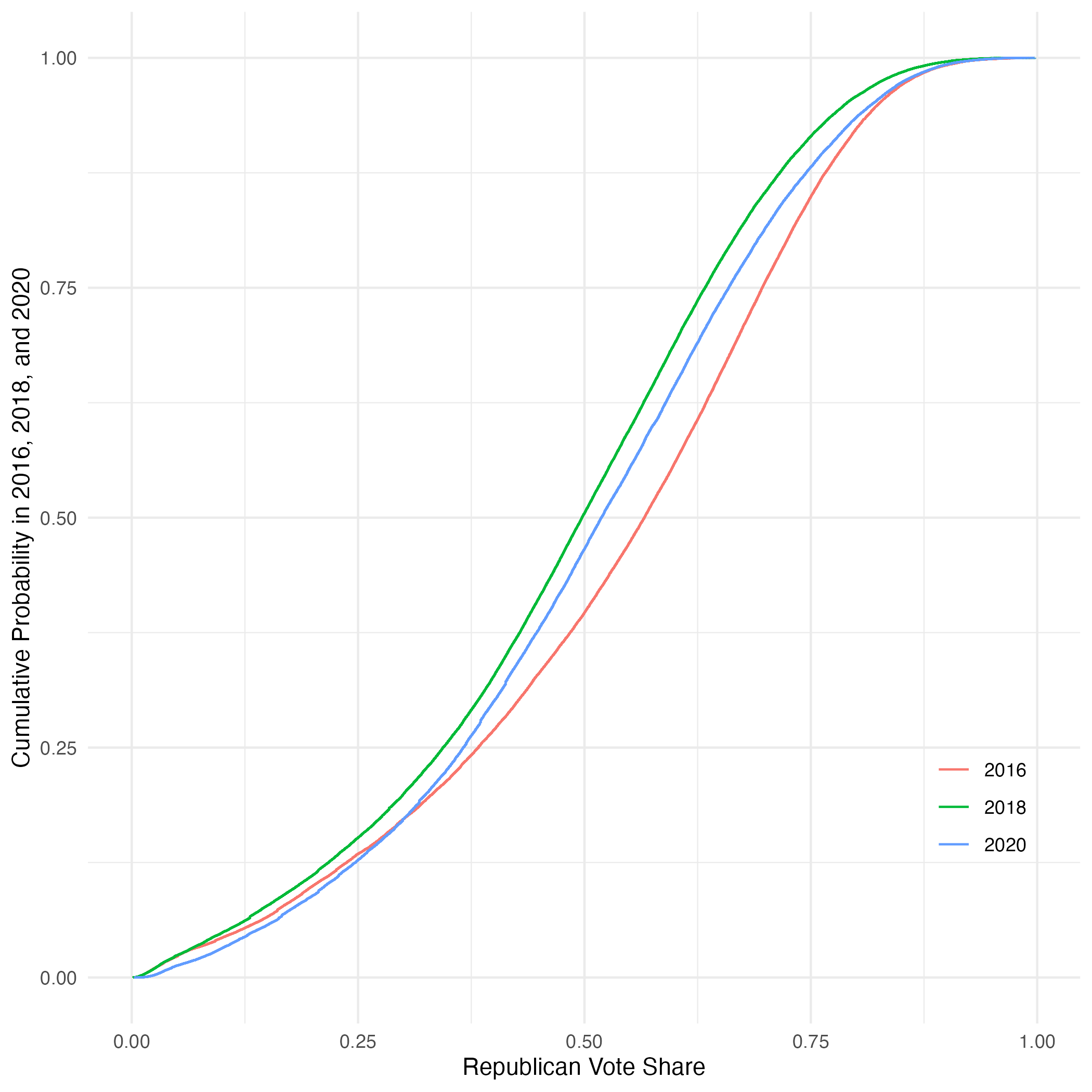}
  \captionsetup{labelformat=empty}
  \caption{(a) CDF}
  \label{fig:CDFVoteShare}	
  \end{subfigure}
 \hfill
  \begin{subfigure}[b]{0.49\textwidth}
  \includegraphics[width=\textwidth]{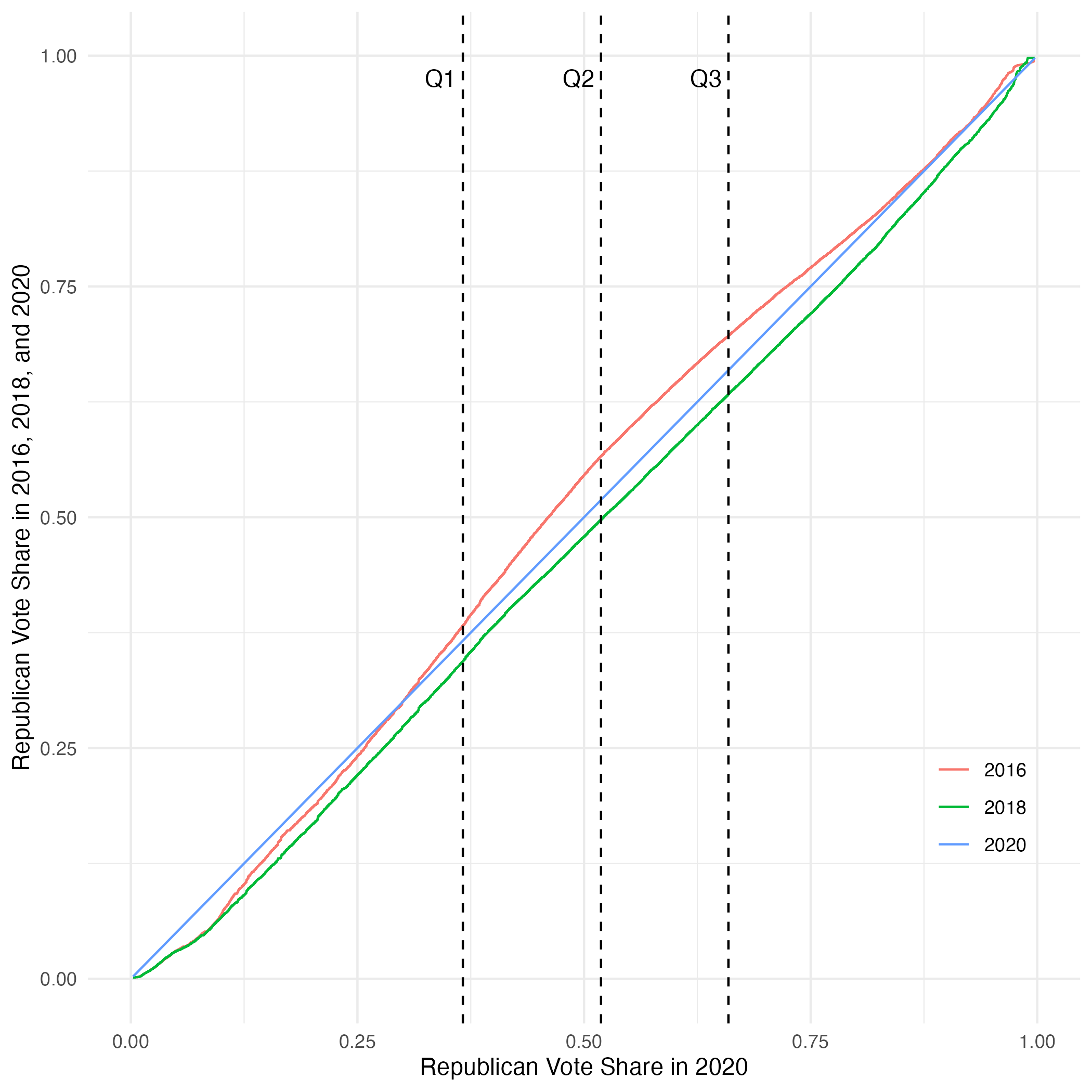}
  \captionsetup{labelformat=empty}
  \caption{(b) Normalized CDF}
  \label{fig:InverseCDF}
  \end{subfigure}  
  \caption{CDF and Normalized CDF of Precinct Vote Shares in 2016, 2018, and 2020}
  \caption*{\emph{Notes:} The left panel displays the empirical cdf of the precinct vote share in 2016, 2018, and 2020, which we denote by $J_t (v)$ for $t \in \{2016,2018,2020\}$. The right panel displays the curves $J^{-1}_{2016} (J_{2020} (v))$, $J^{-1}_{2018} (J_{2020} (v))$, and $J^{-1}_{2020} (J_{2020} (v))=v$, as well as the first, second, and third quantiles of $J_{2020} (v)$.}
\end{figure}

We can use these curves to assess the realism of our key assumption that moderates are swingier than extremists (Assumption 1). Figure \ref{fig:InverseCDF}(b) transforms Figure \ref{fig:CDFVoteShare}(a) by normalizing by the empirical vote-share distribution in 2020. Thus, in Figure \ref{fig:InverseCDF}(b) the blue curve is the $45^\circ$ line; the red curve is the 2016 Republican vote share for a precinct with a given 2020 Republican vote share; and the green curve is the analogous curve for 2018.\footnote{Technically, since we cannot match precincts across elections, the red curve is the 2016 Republican vote share for a precinct \emph{at the same quantile of the vote share distribution} as a precinct with a given 2020 Republic vote share, and similarly for the green curve.} Under our assumptions---including Assumption 1---the red curve should be concave and everywhere above the blue curve, and the green curve should be convex and everywhere below the blue curve, where these concavity/convexity properties reflect Assumption 1. Figure \ref{fig:InverseCDF} shows that this is not exactly true in our data, because the green and red curves are ``too low'' for the left-most districts (a small minority of districts, lying well into the lowest quartile of the vote-share distribution, as indicated in the figure). We believe that this small deviation from Assumption 1 likely reflects an unusually strong performance by Republicans in urban districts in 2020, largely due to a well-documented shift in the Hispanic vote toward Republicans (e.g., \citealt{Igielnik}, \citealt{Kolko}). Such demographic-specific shocks are, of course, outside our model, but could be explored in future work. Overall, we believe Figure \ref{fig:InverseCDF} is well-explained by a combination of our assumptions (including Assumption 1) and an unexpected shift toward Republicans in urban areas in 2020. %

\subsection{Estimates for $\gamma$} We now estimate the key parameter $\gamma$ under the assumption that aggregate and idiosyncratic shocks are both normally distributed. Since districting plans in the US are drawn at the state level, we estimate $\gamma$ separately for each US state. Without loss, we normalize the variance of the taste shock distribution to $1$, so that $Q=\Phi$, the standard normal cdf, and the aggregate shock distribution $G$ is given by a centered normal cdf with standard deviation $\gamma^{-1}$. Recall that our theoretical and numerical results in Section \ref{s:Y} indicate that PMP is optimal if $\gamma\leq 1$, Y-districting is optimal if $\gamma \in (1,1.65)$, and POP is optimal if $\gamma \geq1.65$. Thus, a key question of interest is which of these three regions contains our estimate of $\gamma$.

We estimate $\gamma$ by method of moments. Recall that $v_{nt}$ is the Republican share of the two-party vote in precinct $n$ and election $t$. Let $w_{nt}=\Phi^{-1}(v_{nt})$, the corresponding quantile of the standard normal distribution. Next, define
\[w_t = \frac{\sum_n k_{nt}w_{nt}}{\sum_n k_{nt}} \quad \text{and} \quad w=\frac{\sum_t w_t}{T},\] where the sums over $n$ range over all precincts in a given state. Thus, $w_t$ is the average value of $w_{nt}$  over precincts in the state, weighted by the number of votes in each precinct; and $w$ is the average value of $w_t$ over elections $t$. It is then easy to show that an unbiased and consistent estimator of $\gamma$ is given by
\[\wh{\gamma}={1}\bigg/ {\sqrt{\frac{\sum_t (w_t -w)^2}{T-1}} },\]
and, for any $\alpha \in (0,1)$, a $1-\alpha$ confidence interval for $\gamma$ is given by
\[
\sqrt{\frac{\chi^2_{T-1}(\alpha/2)}{T-1}}\wh \gamma \leq \gamma\leq \sqrt{\frac{\chi^2_{T-1}(1-\alpha/2)}{T-1}}\wh \gamma.
\]

\begin{figure}
  \centering
  \includegraphics[width=0.8\textwidth]{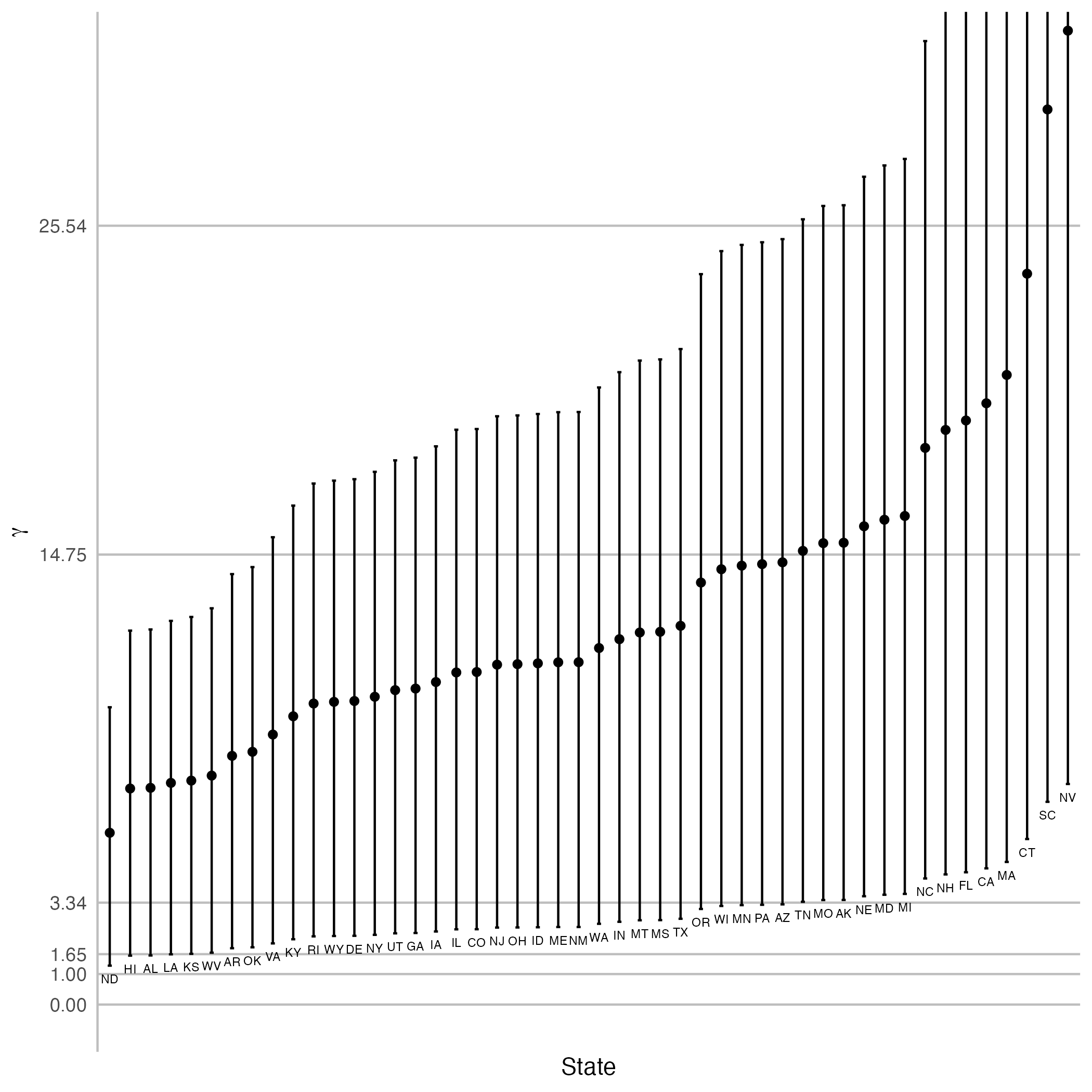}
  \caption{Point Estimates and 90\% Confidence Intervals for $\gamma$ by State}
    \caption*{\emph{Notes:} POP is optimal for $\gamma\geq 1.65$. The point estimate and 90\% confidence interval for $\gamma$  for the entire US are 14.75 and [3.34, 25.54].}
  \label{fig:Sigma}
\end{figure}

Figure \ref{fig:Sigma} displays the results of this estimation. The figure shows the 90\% confidence interval for $\gamma$ for each state. The confidence intervals are extremely wide, because we only have data from three elections, i.e., $T=3$. However, it is clear that the central estimates for $\gamma$, as well as the lower bound of the 90\% confidence interval for almost all states, is well above the critical value of 1.65. The lowest estimate for $\gamma$ for any state is 5.63, the mean estimate for $\gamma$ (weighted by the number of districts in each state) is 14.32, and the corresponding estimate when we estimate $\gamma$ for the US as a whole is 14.75. These estimates are all far above the critical value of 1.65. Moreover, even with $T=3$, the lower endpoint of the 90\% confidence interval is above 1.65 for all states except North Dakota (where the lower endpoint is 1.28), Hawaii (1.6), Alabama (1.61), and Louisiana (1.65). We expect that if we expanded our dataset to include the returns from the 2012 and 2014 elections (thus covering all five congressional elections held under the 2010 districting plans), the lower endpoints of the 90\% confidence interval would exceed 1.65 for these states as well.\footnote{Precinct-level returns for 2012 and 2014 have been compiled by \citet{Ansolabehere} but are less complete and less standardized than the \citet{Baltz} data we use, which only cover 2016, 2018, and 2020.} The data thus clearly indicate that $\gamma$ is well above 1.65 in practice, at least for the vast majority of states, and probably for all of them. Together with the results in Section \ref{s:Y}, this provides strong evidence that optimal gerrymandering is given by POP for realistic parameters.\footnote{While it is not relevant for determining the qualitative form of optimal districting, we can also estimate the distribution $F$ of voter types $s$. At the country-level, the mean estimate of $F$ (calculated as $w$) is very close to $0$, and the standard deviation estimate of $F$ (calculated as $\sqrt{\sum_{n,t} k_{nt}(w_{nt}-w_t)^2/\sum_{n,t} k_{nt}}$) is $0.63$. These values are similar to those in Figure \ref{fig:numeric}. Note however that these estimates may be biased by dropping uncontested elections (unlike our estimates of $\gamma$, which remain unbiased after dropping any set of districts). We also note that the correlation between our estimates of $\gamma$ and the standard deviation of $F$ at the state level (weighted by the number of districts in each state) is small ($-.28$), which is consistent with varying $\gamma$ in $G(r):=Q(\gamma r)$ for fixed $Q$ and $F$  as in Figure \ref{fig:numeric}. In contrast, for an alternative normalization with $Q(t):= G(t/\gamma)$ for fixed $G$ and $F$, the weighted correlation between our estimates of $\gamma$ and the standard deviation of $F$ is large (.79), which would be inconsistent with varying $\gamma$ for fixed $G$ and $F$.}

Our estimates for $\gamma$ are so high that not only is POP clearly optimal rather than PMP, but the optimal POP plan is very similar to $p$-segregation, and both $p$-segregation and traditional pack-and-crack districting are approximately optimal. (Recall Figure \ref{fig:numeric}, where POP is already close to $p$-segregation when $\gamma =6$.) This result can rationalize why actual gerrymandered districting plans usually resemble $p$-segregation or traditional pack-and-crack, rather than POP.

\section{Discussion: Why Does the Form of Gerrymandering Matter?} \label{s:discussion}

Gerrymandering has been a major concern in American politics for many years and has been tied to several important political and legal issues. In this section, we briefly discuss potential implications of our results on the form of optimal partisan gerrymandering---in particular, whether gerrymanderers optimally pack opponents or moderates---for some of these broader issues. We focus on two areas: implications for how regulations and restrictions on districting affect partisan representation, and implications for how gerrymandering affects political competition and polarization.

\subsection{Effects of Districting Restrictions on Partisan Representation}

American state and federal election laws have long recognized potential harms associated with gerrymandering and have therefore restricted gerrymandering in various ways. At the federal level, the key laws are the Equal Protection Clause of the Fourteenth Amendment and the Voting Rights Act of 1965. These laws have been interpreted as not only prohibitting adverse racial gerrymandering, but also as affirmatively requiring states to create electoral districts where racial or ethnic minority voters form either a majority (a so-called ``majority-minority district'') or a large enough minority so as to have a strong opportunity to elect their candidate of choice, perhaps in coalition with some majority voters (often called a ``minority opportunity district'') (e.g., \citealt{Canon}). The creation of such districts played a significant role in increasing Black representation in state legislatures and the US Congress from the 1970's onward, especially in the South (\citealt{GrofmanHandley}, \citealt{CH}). However, the overall \emph{partisan} impact of majority-minority and minority opportunity districts has long been a hotly contested issue, with some observers arguing that these districts effectively pack strong Democratic supporters and thus resemble a component of a Republican-optimal districting plan. This issue came to a head following the 1994 Republican takeover of the US House, which many journalists and political scientists blamed in part on the creation of majority-minority districts in the 1990 redistricting cycle; however, other observers have disputed this narrative (see, e.g., \citealt{CH} and references therein, \citealt{Cameron}, \citealt{Washington}).

Following \citet{CH}, we argue that whether a requirement to create majority-minority or minority opportunity districts is likely to increase or decrease overall Republican representation hinges to a large degree on whether optimal partisan gerrymandering packs opponents or moderates. The convential view throughout the 1990's (what Cox and Holden call the ``pack-and-crack consensus'') was that optimal gerrymandering packs opponents, and hence that a requirement to create majority-minority districts that pack strong Democratic supporters may well increase overall Republican representation.\footnote{Minority opportunity districts may or may not raise similar issues, depending on the share of strong Democratic supporters in these districts \citep{Lublin}.} Based on the analysis of \citet{FH}, \citet{CH} challenge this consensus by arguing that optimal districting is given by PMP, and thus packs moderates rather than opponents. Since a PMP plan does not create districts packed with strong Democratic supporters, \citeauthor{CH} argue that a requirement to create such districts precludes PMP and is therefore likely to reduce overall Republican representation.

We agree with \citeauthor{CH} that whether optimal districting packs opponents or moderates is likely to be an important determinant of whether a requirement to create majority-minority or minority opportunity districts increases or decreases overall Republican representation. However, \citeauthor{CH}'s argument that PMP is optimal in practice rests on the implicit assumption that the low-idiosyncratic-uncertainty case studied by \citet{FH} is representative. For example, \citeauthor{CH} write, ``In a world with diverse voter types, however, there is no plausible distribution of African American voters that would make it optimal for Republican redistricting authorities to create districts in which African Americans make up a supermajority of voters. Within the model, packing one's opponents is never the optimal strategy,'' (p. 574). Our results instead indicate that, empirically, idiosyncratic uncertainty is much larger than aggregate uncertainty, and that in this case POP is optimal (and traditional pack-and-crack districting is approximately optimal), so Republicans do benefit from packing strong Democratic voters. Thus, by analyzing a general model that allows diverse voter types but does not restrict the relative amounts of idiosyncratic and aggregate uncertainty, we can let the data determine which form of districting plan is optimal in practice, and we find that POP is optimal for realistic parameters. Overall, our results support the traditional ``pack-and-crack consensus''---Republicans benefit from packing strong Democratic voters---over \citeauthor{CH}'s challenge based on the optimality of packing moderates for certain parameter values.

Of course, even if POP is optimal, so that packing strong Democratic voters \emph{in the Republican-optimal manner} benefits Republicans, whether a requirement to create majority-minority or minority opportunity districts benefits Republicans in practice is an empirical question. A requirement to create a large numbers of districts with relatively small Democratic majorities can obviously hurt Republicans. Moreover, as emphasized by \citet{Shotts01}, \emph{any} constraint on districting weakly hurts Republicans in states where Republicans control districting. In general, we believe that understanding the form of partisan-optimal unconstrained districting is useful for assessing the likely impact of restrictions on districting, such as those imposed by Voting Rights Act, but as a complement to empirical analysis rather than a substitute.

\subsection{Effects of Gerrymandering on Political Competition and Polarization}

A second area of debate concerns the impact of gerrymandering on the intensity of electoral competition (e.g., the fraction of ``competitive'' districts or the extent of incumbency advantage) and political polarization. Popular discourse often blames gerrymandering for reducing competition and increasing polarization. While the scholarly literature is generally skeptical of the claim that gerrymandering plays a large role in explaining overall secular trends in competition and polarization (e.g., \citealt{GelmanKing}, \citealt{Abramowitz}, \citealt{MPR}, \citealt{FH09}), some work does find such effects (e.g., \citealt{Cottrell}, \citealt{Kenny}), and the issue remains contested.

Regardless of the overall effects of gerrymandering on competition and polarization, the nature of these effects likely depends on the form that gerrymandering takes. Roughly speaking, with a right-wing designer, POP (as well as $p$-segregation and traditional pack-and-crack) create a few strongly left-leaning districts and many slightly right-leaning districts, with a ``gap'' between the left-leaning and right-leaning districts. Formally, under POP, there is always a gap between the highest value of $r^* (P)$ for a district $P$ in the interval of segregated voter types and the lowest value of $r^* (P)$ for a district $P$ in the interval of paired types (see, e.g., the last three panels in Figure \ref{fig:numeric}). POP also involves relatively low polarization within each district, since the lowest voter types in cracked districts are ``moderates'' rather than extreme left-wingers. In contrast, PMP creates a continuum of districts ranging from left-leaning to right-leaning---formally, the set $\{ r:r=r^* (P) \text{ for some } P\in \supp (\mathcal{H}) \}$ is an interval (see, e.g., the first three panels in Figure \ref{fig:numeric})---with less extreme left-leaning districts than under POP. PMP also involves greater within-district polarization than POP, at least in the sense that the maximum range of voter types that are pooled together under PMP is greater than under POP (since this range is as large as possible under PMP, but is strictly smaller under POP).

Our model does not encompass any endogenous political responses to districting, such as effects of districting on which politicians run for office and on what platforms. With this caveat in mind, we can draw some tentative implications of the above features of POP (or $p$-segregation or traditional pack-and-crack) and PMP for political competition and polarization. First, the fact that the distribution of threshold shocks $r^* (P)$ has a gap under POP but not under PMP suggests that pack-and-crack plans may lead to a more polarized legislature, where the packed districts elect left-wing representatives, and the cracked districts elect right-leaning representatives. The possibility that packing opponents can increase polarization in this manner is a long-standing political and legal concern (see, e.g., \citealt{CH}, p. 595). \citet{CK}, \citet{BesleyPreston}, and \citet{Bracco} develop models with this feature. In contrast, PMP may lead to a less polarized legislature. Second, POP may lead to a larger number of ``uncompetitive,'' far-left districts. Creating uncompetitive districts is usually viewed as a socially undesirable feature of a districting plan, but see \citet{Buchler} and \citet{Brunell} for opposing views. Finally, lower within-district polarization under POP may be socially desirable if voters benefit from being ideologically close to their representative, as in \citet{BesleyPreston} and \citet{GPS}. These and other implications of optimal districting for political processes and outcomes could be studied more fully in a model that endogenized additional aspects of political competition beyond districting. This is a promising direction for future research.

\section{Conclusion} \label{s:conclusion}

This paper has developed a simple and general model of optimal partisan gerrymandering. Our main message has four parts. First, pack-and-pair districting---a generalization of traditional packing-and-cracking---is typically optimal for the gerrymanderer. Second, the optimal form of pack-and-pair depends on the relative amounts of aggregate and idiosyncratic uncertainty facing the gerrymanderer: opposing voters are packed when idiosyncratic uncertainty dominates, while moderate voters are packed when aggregate uncertainty dominates. Third, empirically, idiosyncratic uncertainty dominates in practice, so we expect pack-opponents-and-pair (POP) districting to be optimal. This finding also establishes that the relevant parameter range for future research on gerrymandering (and electoral competition more generally) is that where idiosyncratic uncertainty is much larger than aggregate uncertainty. Fourth, estimated idiosyncratic uncertainty is so large that the optimal POP plan closely resembles a simpler pack-opponents-and-pool plan, where more favorable voters are all pooled together, rather than being paired as they are under POP; moreover, traditional pack-and-crack districting, where less favorable voters are also all pooled together, rather than being segregated, is also approximately optimal. This final observation can rationalize the use of traditional pack-and-crack districting plans in practice.

Methodologically, we develop and exploit a tight connection between gerrymandering and information design. We show that a general model of partisan gerrymandering is equivalent to a general Bayesian persuasion problem where the state of the world and the receiver's action are both one-dimensional and the sender's preferences are state-independent. This common framework nests the important prior contributions of \citet{OG}, \citet{FH}, and \citet{GP}, and facilitates a more general and realistic analysis that allows diverse voter types and non-linear vote swings without restricting the relative amounts of aggregate and idiosyncratic uncertainty.

We hope our model can inform future research on various aspects of redistricting. We mention a few directions for future research.

First, we have assumed that the designer maximizes his party's expected seat share. It may be more realistic to assume that the designer's utility is non-linear in his party's seat share, for example because he puts a premium on winning a majority of seats. We examined this case in an earlier version of the current paper \citep{KW}. While non-linear designer utility introduces some new complications, it also reinforces the main message of the current paper, in that if the designer's utility is S-shaped in his party's seat share (as in the case with a premium on winning a majority), then pack-opponents-and-pool is strictly optimal even with linear swing and uniform aggregate shocks (whereas a designer with linear utility is indifferent among all districting plans in this case).

Second, we have assumed that all voters always vote, or at least always vote at the same rate (as is equivalent). It would be interesting to incorporate heterogeneous turnout in the analysis. A recent contribution by \citet{Bouton} considers voters with a binary partisan type (as in \citealt{OG}) and a continuous ``turnout type,'' which captures fixed turnout heterogeneity across voters. An alternative model, which captures variable turnout heterogeneity, would retain one-dimensional voter types but assume that voters abstain when they are close to indifferent between the parties. It would be interesting to compare these models, as in practice turnout heterogeneity has both fixed sources (e.g., education, race) and variable ones (e.g., almost-indifferent voters turn out less).

Third, a robust prediction of our analysis is that there should be greater within-district polarization in districts that are more favorable for the designer's party. It would be interesting to test this prediction empirically.

Further questions include, what does the model imply for political competition and the resulting policy choices? What are the model's comparative statics---for example, what factors determine the proportion of packed and cracked districts?\footnote{\citet{KW} analyze comparative statics with binary voter types.} What does the model imply about how gerrymandering should be measured and regulated? A better understanding of the form of optimal partisan gerrymandering can contribute to the study of these questions and related ones.

\begingroup
\setstretch{1.0}

\bibliographystyle{econometrica}
\bibliography{persuasionlit}
\endgroup

\appendix

\section*{\Large{Appendix: Proofs}}

Given the equivalence between our model and a class of Bayesian persuasion problems described in Section \ref{s:model}, Propositions \ref{packandcrack}, \ref{nounc}, \ref{p:linearW}, and \ref{p:sdd} follow from prior results in the persuasion literature. For these results, we give references to the literature as well as (mostly) self-contained proofs, for completeness. In contrast, Propositions \ref{nohet}, \ref{p:swingy}, and \ref{p:pp}--\ref{p:gamma} are new to both the persuasion and gerrymandering literatures. We give complete proofs of these results.

\section{Proofs for Section \ref{s:benchmark}} \label{a:A}

\begin{proof}[Proof of Proposition \ref{packandcrack}] This result is standard (see, e.g., Figure 1 in \citealt{OG}). Case (1) is trivial, as the designer wins all districts if he creates measure $1$ of districts satisfying $\Pr_P(s\geq r^0)\geq 1/2$ and loses a positive measure of districts otherwise. For case (2), note that since the designer wins a district $P$ iff $\Pr_P(s\geq r^0)\geq 1/2$, a districting plan can be described by a distribution $H$ over $x=\Pr_P(s\geq r^0)$.  The designer's utility for any feasible $H$ is
\begin{equation}\label{e:pc2}
\int \1 \left\{x\geq \tfrac 12\right \}dH(x)\leq \int 2xdH(x)=2m,
\end{equation}
where the inequality holds because $\1 \{x\geq 1/2\}\leq 2x$ for all $x\in [0,1]$, and the equality holds because $\int xdH(x)=m$ for any feasible $H$, by the law of iterated expectations. Thus, any plan that creates measure $2m$ of cracked districts satisfying $\Pr_P(s\geq r^0)=1/2$ and measure $1-2m$ of packed districts satisfying $\Pr_P(s<r^0)=1$ is optimal. Moreover, any other plan creates a positive measure of districts with $\Pr_P(s\geq r^0)\notin \{0,1/2\}$ (i.e., $\supp (H)\nsubseteq \{0,1/2\}$), so that the inequality in \eqref{e:pc2} is strict, because $\1 \{x\geq 1/2\}= 2x$ iff $x\in \{0,1/2\}$. So any such plan is suboptimal.
\end{proof}

\begin{proof}[Proof of Proposition \ref{nounc}] 
The proposition can be obtained using the proofs of Lemmas 1 and C1 in \citet{Kolotilin2014b}. Case (1) is trivial, as the designer wins all districts if he creates measure $1$ of districts satisfying $\int v(s,r^0)dP(s)\geq 1/2$ and loses a positive measure of districts otherwise. For case (2), note that since $v(s,r^0)$ is differentiable and strictly increasing in $s$, we can redefine $s$ as $v(s,r^0)$, so that the redefined $s$ has a full-support density on  $[\ul s,\ol s]$, with $0\leq \ul s<\ol s\leq 1$. Assume that $\ol s>1/2$, as otherwise the result is trivial. Since $\E_F [s]< 1/2$, there is a unique $s^*\in (\ul s,\ol s)$ satisfying $\E_F[s|s\geq s^*]=1/2$. Define
\[
\ol U(x)= 
\begin{cases}
0, & x< s^*,\\
\frac{x-s^*}{1-2s^*}, &x\geq s^*.
\end{cases}
\]
Since the designer wins a district $P$ iff $\E_P [s]\geq 1/2$, his expected seat share under a plan $\mathcal H$ is
\begin{equation}\label{e:nounc}
\begin{gathered}
\int \1 \left\{\E_P[s]\geq \tfrac 12\right \}d\mathcal H(P)\leq \int \ol U(\E_P[s])d\mathcal H(P)\leq\iint \ol U(s)dP(s)d\mathcal H(P)\\
= \int \ol U(s)d F(s)= \int_{s^*}^{\ol s}\frac{s-s^*}{1-2s^*}dF(s)= 1-F(s^*),
\end{gathered}
\end{equation}
where the first inequality holds because $\1 \{x\geq 1/2\}\leq \ol U(x)$ for all $x$, the second inequality holds because $\ol U$ is convex, the first equality holds because $\int P dH(P)=F$, and the last equality holds because $\E_F[s|s\geq s^*]=1/2$. Thus, a plan $\mathcal H$ is optimal iff for all $P\in \supp(\mathcal H)$ we have: (i) $\E_P[s]\leq s^*$ or $\E_P[s]=1/2$ (as otherwise the first inequality in \eqref{e:nounc} is strict), and (ii) $\supp (P)\subset [\ul s,s^*]$ if $\E_P[s]< s^*$ and $\supp (P)\subset [s^*,\ol s]$ if $\E_P[s]=1/2$  (as otherwise the second inequality in \eqref{e:nounc} is strict). This means that $\mathcal H$ contains measure $F(s^*)$ of districts $P$ where $\Pr_P(s<s^*)=1$ and measure  $1-F(s^*)$ of districts $P$ where $\Pr_P (s\geq s^*)=1$ and $\E_P[s]=1/2.$
\end{proof}

\begin{proof}[Proof of Proposition \ref{nohet}]
For a districting plan $\mathcal H$, define $H$ as $H(r) = \Pr_{\mathcal H}(r^*(P)\leq r)$ for all $r$. The designer thus wins measure $1-H(r_-)$ of districts when the realized aggregate shock is $r$. For each realization $r$, the designer wins a district $P$ iff it contains at least measure $1/2$ voters with types $s\geq r$ (i.e., $\Pr_P (s\geq r)\geq 1/2$). Since the population has measure $1-F(r)$ voters with types $s\geq r$, the designer wins at most measure $2(1-F(r))$ districts, so $1-H(r_-)\leq 2(1-F(r))$. Since the designer can win at most measure $1$ districts, any feasible $H$ satisfies $H(r_-)\geq H^* (r)$, where
\[
H^*(r)=
\begin{cases}
0, &\text{if $r\leq s^m$},\\
1-2(1-F(r)), &\text{if $r> s^m$}.
\end{cases}
\]
Thus, the designer's expected seat share for any feasible $H$  is
\begin{equation*}
\int \left(1-H(r_-)\right)dG(r) \leq \int \left(1-H^*(r)\right)dG(r),
\end{equation*}
with strict inequality if $H(r_-)> H^*(r)$ for some $r$ (and thus on some interval $(r,r')$ with $r'>r$, by continuity of $H^*$ and monotonicity of $H$), because $G(r)$ is strictly increasing in $r$. Thus, a districting plan $\mathcal H$ is optimal iff it induces $H^*$, which means that $\mathcal H-$almost every district $P$ that the designer wins iff the aggregate shock is at most $r$ satisfies $\Pr_P(s=r)=\Pr_P(s< s^m)=1/2$.	
\end{proof}

\begin{proof}[Proof of Proposition \ref{p:linearW}]
The proposition follows from Theorem 1 in \citet{KMZ}. The proof is similar to the proof of Proposition \ref{nounc}. 
The most interesting case is where there is an interior cutoff $s^*$ and pool mean $x^*=\E_F[s\geq s^*]$ satisfying $u(x^*)(x^* - s^*)=U(x^*)-U(s^*)$. As follows from Figure \ref{f:sp}, such $s^*$ is unique. Define
\[
\ol U(x)= 
\begin{cases}
U(x), & x< s^*,\\
U(x^*)+u(x^*) (x-x^*), &x\geq s^*.
\end{cases}
\]
The designer's expected seat share under a plan $\mathcal H$ is
\begin{equation}\label{e:linearW}
\begin{gathered}
\int U(\E_P[s]) d\mathcal H(P)\leq \int \ol U(\E_P[s])d\mathcal H(P)\leq\iint \ol U(s)dP(s)d\mathcal H(P)\\
= \int \ol U(s)d F(s)=\int_{0}^{s^*}U(x)dF(x)+U(x^*)(1-F(s^*)),
\end{gathered}
\end{equation}
where the first inequality holds by $U\leq \ol U$, the second inequality holds by convexity of $\ol U$, the first equality holds by $\int P dH(P)=F$, and the second equality holds by the definition of $s^*$, $x^*$, and $\ol U$. Thus, a plan $\mathcal H$ is optimal iff for all $P\in \supp (\mathcal H)$ we have: (i) $\E_P[s]\leq s^*$ or $\E_P[s]=x^*$ (as otherwise the first inequality in \eqref{e:linearW} is strict), and (ii) $P=\delta_{\E_P[s]}$ if $\E_P[s]< s^*$ and $\supp (P)\subset [s^*,\ol s]$ if $\E_P[s]=x^*$ (as otherwise the second inequality in \eqref{e:linearW} is strict). This implies that the distribution of district means induced by pack-opponents-and-pool districting with cutoff $s^*$ is uniquely optimal.
\end{proof}

\section{Proofs for Section \ref{s:nonlinear}} \label{a:B}
We start with a lemma that distills some key results from \citet{KCW}.
\begin{lemma}\label{l:dual}
There exists a bounded, measurable function $\lambda : \R\rightarrow \R $ such that, for any optimal districting plan $\mathcal H$, the following hold:
\begin{enumerate}
	\item For all $P,P'\in \supp (\mathcal H)$ and all $s\in \supp (P)$, we have
	\[
	G(r^*(P))+\lambda (r^*(P)) \left(v(s,r^*(P))-\tfrac 12\right)\geq G(r^*(P'))+\lambda (r^*(P')) \left(v(s,r^*(P'))-\tfrac 12\right).
	\]
	\item For all $P\in \supp (\mathcal H)$, we have
	\[
	\lambda  (r^*(P)) = - \frac{g(r^*(P))}{{\int} \frac{\partial v(s,r^*(P))}{\partial r}d P(s)}.
	\]
	\item For any non-degenerate $P\in \supp (\mathcal H)$, $\lambda$ has a derivative $\lambda '(r^*(P))$ at $r^*(P)$ satisfying, for all $s\in \supp (P)$,
	\[
	g(r^*(P))+\lambda (r^*(P)) \frac{\partial v(s,r^*(P))}{\partial r} +\lambda' (r^*(P)) \left ( v(s,r^*(P))-\tfrac 12 \right)=0.
	\]
\end{enumerate} 
\end{lemma}

Intuitively, $\lambda(r^*(P))$ is the multiplier on the constraint $\int v(s,r^*(P))dP=\frac{1}{2}$. Part 2 of the lemma says that $\lambda(r^*(P))$ equals the product of the designer's marginal utility of increasing $r^*(P)$ (which equals $g(r^*(P))$) and the rate at which $r^*(P)$ increases as the constraint $\int v(s,r^*(P))dP=\frac{1}{2}$ is relaxed (which equals $-1/{\int} \frac{\partial v(s,r^*(P))}{\partial r}d P(s)$ by the implicit function theorem). Part 1 of the lemma says that the designer assigns a type-$s$ voter to a district $P$ so as to maximize $G(r^*(P))+\lambda (r^*(P)) \left(v(s,r^*(P))-\tfrac 12\right)$. Part 3 says that the first-order condition of this maximization problem with respect to $r$ holds for all non-degenerate $P \in \supp(\mathcal{H})$ and all $s \in \supp(P)$.

\begin{proof}
Any districting plan $\mathcal H$ induces a joint distribution $\pi_{\mathcal H}$ of voter type $s$ and the threshold aggregate shock $r$ below which the designer wins a district containing voter type $s$. Specifically, denoting $\ul r=r^*(\delta_{\ul s})$ and $\ol r=r^*(\delta_{\ol s})$, $\mathcal H$ induces $\pi_{\mathcal H}$ given by
\[\pi_{\mathcal H}(S,R):= \int_{P:r^*(P)\in R} P(S) d\mathcal H(P)\quad \text{for all measurable $S\subset [\ul s,\ol s]$ and $R\subset [\ul r, \ol r]$}.\]
Appendix B in \citet{KCW} constructs a suitable bounded, measurable function $\lambda : [\ul r, \ol r]\rightarrow \R $, and defines the set $\Gamma$ as
\[
\Gamma: = \{(s,r)\in [\ul s,\ol s]\times [\ul r, \ol r]: \sup_{\tilde r\in [\ul r,\ol r]} \{G(\tilde r)+\lambda(\tilde r)\left(v(s,\tilde r)-\tfrac 12\right)\}=G(r)+\lambda( r)\left(v(s, r)-\tfrac 12\right)\}.
\]
Moreover, they define 
\begin{align*}
R_\Gamma &:=\{r\in [\ul r,\ol r]:(s,r)\in \Gamma\quad \text{ for some }s\in [\ul s,\ol s]\},\\
\Gamma_r &:=\{s\in [\ul s,\ol s]:(s,r)\in \Gamma\}\quad \text{for all $r\in [\ul r, \ol r]$}.
\end{align*}
Part 1 of their Theorem 7 shows that the set $\Gamma$ is compact and satisfies 
\begin{equation}\label{e:Gamma}
\min \Gamma_r\leq s^*(r)\leq \max  \Gamma_r \quad  \text{for all $r\in R_{\Gamma}$},	
\end{equation} 
where $s^*(r)$ is defined by $v(s^*(r),r)=1/2$. Moreover, the same result shows that 
\begin{equation}\label{e:suppGamma}
\supp(\pi_{\mathcal H})\subset \Gamma\quad  \text{for each optimal $\mathcal H$}.	
\end{equation} 

Furthermore, \citeauthor{KCW} define the set $\Gamma^*\subset \Gamma$ such that
\[
\Gamma^*_{r}=
\begin{cases}
	\{s^*(r)\}, &r\in R_\Gamma\text{ and }s^*(r) \in \{\min \Gamma_r, \max \Gamma_r\},\\
	\Gamma_r, &\text{otherwise},
\end{cases}
\quad \text{for all $r\in [\ul r, \ol r]$}.
\]
Part 2 of their Theorem 7 shows that, if $\Gamma^*_r=\{s^*(r)\}$, then
\begin{equation}\label{e:min=max}
g(r)+\lambda(r)\frac{\partial v(s^*(r),r)}{\partial r}=0,	
\end{equation}
and if $\min \Gamma^*_r<s^*(r)< \max \Gamma^*_r$, then $\lambda$ has a derivative $\lambda'(r)$ at $r$ satisfying, for all $s\in \Gamma^*_r$,
\begin{equation}\label{e:min<max}
g(r)+\lambda(r)\frac{\partial v(s,r)}{\partial r}+\lambda'(r) \left(v(s,r)-\tfrac 12\right)=0.	
\end{equation}

Now, consider any optimal $\mathcal H$. By \eqref{e:suppGamma}, we have $\supp(P)\subset \Gamma_{r^*(P)}$ for all $P\in \supp(\mathcal H$). By the definition of $r^*(P)$, we have $\int v(s,r^*(P))dP(s)=1/2$, so either $\supp(P)=\{s^*(r^*(P))\}$ or $\min \supp (P)<s^*(r^*(P))<\max \supp (P)$. In both cases, we have $\supp (P)\subset \Gamma^*_{r^*(P)}$, by \eqref{e:Gamma} and the definition of $\Gamma^*_{r^*(P)}$. Thus, part 1 of the lemma follows from the definition of $\Gamma$. In turn, part 2 follows from \eqref{e:min=max} when $P$ is degenerate and from integrating \eqref{e:min<max} over $P$ when $P$ is non-degenerate. Finally, part 3 follows from \eqref{e:min<max}.
\end{proof}

\begin{proof}[Proof of Proposition \ref{p:swingy}]
Part 1 follows from \eqref{swing} and $v(s,r)=Q(s-r)$. For part 2, notice that \eqref{swing} is equivalent to
\[
\frac{\partial^3 v(s,r)}{\partial s^2\partial r} \frac {\partial v(s,r)}{\partial s}> \frac{\partial ^2 v(s,r)}{\partial s \partial r} \frac{\partial^2 v(s,r)}{\partial s^2}  \quad \text{for all $s$, $r$}.
\]
Thus, letting subscripts denote partial derivatives, $v_{sr}(s,r)=0$ implies $v_{ssr}(s,r)>0$, so $v_{sr}(s,r)=0$ implies $v_{sr}(s',r)>0$ for all $s'>s$, showing that $v_{sr}(s,r)$ satisfies strict single crossing in $s$, and hence $v_r(s,r)$ is strictly quasi-convex in $s$.
\end{proof}

\begin{proof}[Proof of Proposition \ref{p:sdd}] 
The proposition follows from Theorem 3 in \citet{KCW} for the state-independent sender case, where $V(a,\theta)=V(a)$. We illustrate the proof in the case where $\supp(F)$ and $\supp (\mathcal H)$ are finite. The general proof has the same logic but involves additional technicalities, which can be handled using Lemma \ref{l:dual}. The proof rests on two lemmas.
\begin{lemma}\label{L1}
For any optimal $\mathcal H$ (with finite support), there do not exist $P,P'\in \supp (\mathcal H)$ such that $P$ contains types $s<s''$, $P'$ contains a type $s'\in (s,s'')$, and $r^*(P)<r^*(P').$
\end{lemma}
\begin{proof}
Suppose for contradiction that such districts $P$ and $P'$ exist, and denote $r^*(P)=r$ and $r^*(P')=r'$, with $r<r'$. Consider a perturbation that shifts mass $\rho=(v(s'',r)-v(s',r))\varepsilon $ of type-$s$ voters and mass $\rho''=(v(s',r)-v(s,r))\varepsilon$ of type-$s''$ voters from  $P$ to  $P'$, and shifts an equal mass $\rho'=\rho+\rho''=(v(s'',r)-v(s,r))\varepsilon$ of type-$s'$ from $P'$ to $P$, for a sufficiently small $\varepsilon>0$. Since $v(s,r)$ is strictly increasing in $s$, these masses are strictly positive and thus this perturbation is well-defined. Since the perturbation does not change the mass of voters in $P$ and $P'$, to show that it strictly increases the designer's expected seat share, it suffices to show that $r^*(P)$ does not change and $r^*(P')$ strictly increases. First, $r^*(P)$ does not change because $\int v(s,r)dP(s)$ does not change, as
\[
-v(s,r)\rho + v(s',r)\rho'  -v(s'',r)\rho''=0.
\]
Second,  $r^*(P')$ strictly increases because $\int v(s,r')dP'(s)$ strictly increases, as
\begin{gather*}
v(s,r')\rho -v(s',r')\rho' +v(s'',r')\rho'' \\
=[(v(s'',r')-v(s',r'))(v(s',r)-v(s,r))-(v(s'',r)-v(s',r))(v(s',r')-v(s,r'))]\varepsilon\\
=\left[\int_{s'}^{s''}\int_{s}^{s'}\frac{\partial v(\tilde s',r')}{\partial s}\frac{\partial v(\tilde s,r)}{\partial s} d \tilde s d\tilde s' - \int_{s'}^{s''}\int_{s}^{s'}\frac{\partial v(\tilde s',r)}{\partial s}\frac{\partial v(\tilde s,r')}{\partial s} d \tilde s d\tilde s' \right]\varepsilon\\
=\left[\int_{s'}^{s''}\int_{s}^{s'}\left(\frac{\partial v(\tilde s',r')}{\partial s}\frac{\partial v(\tilde s,r)}{\partial s}-\frac{\partial v(\tilde s',r)}{\partial s}\frac{\partial v(\tilde s,r')}{\partial s}\right)d \tilde s d\tilde s'\right]\varepsilon>0,
\end{gather*}
where the inequality holds because the integrand is strictly positive for $r<r'$ and $\tilde s < \tilde s'$ by Assumption 1.
\end{proof}
\begin{lemma}\label{L2}
For any optimal $\mathcal H$ (with finite support) and any $P\in \supp (\mathcal H)$, we have $|\supp(P)|\leq 2$.
\end{lemma}
\begin{proof}
Suppose for contradiction that there exists a district $P\in \supp (\mathcal H)$ that contains three types $s<s'<s''$. Denote $r^*(P)=r$. Suppose we split district $P$ into two identical equal-sized districts $P'$ and $P''$. Then consider a perturbation that shifts mass $\rho=(v(s'',r)-v(s',r))\varepsilon $ of type-$s$ voters and mass $\rho''=(v(s',r)-v(s,r))\varepsilon$ of type-$s''$ voters from  $P'$ to  $P''$, and shifts an equal mass $\rho'=\rho+\rho''=(v(s'',r)-v(s,r))\varepsilon$ of type-$s'$ voters from $P''$ to $P'$, for a sufficiently small $\varepsilon>0$. Notice that $r^*(P'')=r^*(P')=r$, because
\[
v(s,r)\rho -v(s',r)\rho' +v(s'',r) \rho'' =0.
\] 
Now consider an additional perturbation that moves an infinitesimal mass $d\rho $ of type-$s$ voters from $P''$ to $P'$ and moves the same mass $d\rho $ of type-$s''$ voters from $P'$ to $P''$. By the implicit function theorem, $r^*(P'')=r+dr''+o(dr'')$ and $r^*(P')=r-dr'+o(dr')$, where
\[
dr''= \frac{(v(s'',r)-v(s,r))}{-\int \frac{\partial v(\tilde s,r)}{\partial r} dP''(\tilde s)}dm\quad\text{and}\quad dr'= -\frac{(v(s'',r)-v(s,r))}{-\int \frac{\partial v(\tilde s,r)}{\partial r} dP'(\tilde s)}dm.
\]
To show that this perturbation strictly increases the designer's expected seat share, it suffices to show that $dr''>dr'$, or equivalently $-\int \frac{\partial v(\tilde s,r)}{\partial r} dP''(\tilde s)<-\int \frac{\partial v(\tilde s,r)}{\partial r} dP'(\tilde s)$. This holds because
\begin{gather*}
-\frac{\partial v(s,r)}{\partial r}\rho +\frac{\partial v(s',r)}{\partial r}\rho ' -\frac{\partial v(s'',r)}{\partial r}\rho ''\\
=\left[-\tfrac{\partial v(s,r)}{\partial r} (v(s'',r)-v(s',r))+\tfrac{\partial v(s',r)}{\partial r}(v(s'',r)-v(s,r)) -\tfrac{\partial v(s'',r)}{\partial r}(v(s',r)-v(s,r))	\right]\varepsilon\\
=\left[\left(\tfrac{\partial v(s',r)}{\partial r}-\tfrac{\partial v(s,r)}{\partial r}\right)(v(s'',r)-v(s',r))-\left(\tfrac{\partial v(s'',r)}{\partial r}-\tfrac{\partial v(s',r)}{\partial r}\right)(v(s',r)-v(s,r))\right]\varepsilon\\
=\left[\int _{s}^{s'}\frac{\partial^2 v(\tilde s,r)}{\partial s\partial r}d\tilde s \int_{s'}^{s''}\frac{\partial v(\tilde s',r)}{\partial s}d\tilde s'-\int _{s'}^{s''}\frac{\partial^2 v(\tilde s',r)}{\partial s\partial r}d\tilde s'\int_{s}^{s'}\frac{\partial v(\tilde s,r)}{\partial s}d\tilde s\right]\varepsilon\\
<\frac{\frac{\partial^2 v(s',r)}{\partial s\partial r}}{\frac{\partial v(s',r)}{\partial s}}\left[\int _{s}^{s'}\frac{\partial v(\tilde s,r)}{\partial s}d\tilde s \int_{s'}^{s''}\frac{\partial v(\tilde s',r)}{\partial s}d\tilde s'-\int _{s'}^{s''}\frac{\partial v(\tilde s',r)}{\partial s}d\tilde s'\int_{s}^{s'}\frac{\partial v(\tilde s,r)}{\partial s}d\tilde s\right]\varepsilon=0,
\end{gather*}
where the inequality follows from Assumption 1, which implies that $\partial\ln (\partial v(s,r)/\partial s)/\partial r$ is strictly increasing in $s$, and thus
\[
\frac{\frac{\partial^2 v(\tilde s,r)}{\partial s\partial r}}{\frac{\partial v(\tilde s,r)}{\partial s}}<\frac{\frac{\partial^2 v(s',r)}{\partial s\partial r}}{\frac{\partial v(s',r)}{\partial s}}<\frac{\frac{\partial^2 v(\tilde s',r)}{\partial s\partial r}}{\frac{\partial v(\tilde s',r)}{\partial s}}\quad \text{for $\tilde s <s'<\tilde s'$}.\qedhere
\]

\end{proof}
By Lemmas \ref{L1} and \ref{L2}, to show that every optimal districting plan $\mathcal H$ (with finite support) is single-dipped, it suffices to show that for any district $P \in \supp(\mathcal H)$ consisting of voter types $s<s''$ and any district $P' \in \supp(\mathcal H)$ containing a voter type $s'\in (s,s'')$, we have $r^*(P)\neq r^*(P')$. But this follows because, if $r^*(P)= r^*(P')$, then merging districts $P$ and $P'$ into one district would also be optimal, but the merged district would contain three voter types, contradicting Lemma \ref{L2}.
\end{proof}

\begin{proof}[Proof of Proposition \ref{p:pp}]
Let $\mathcal H$ be a pack-and-pair districting plan. 
Since $\mathcal H$ is strictly single-dipped, the support of each $P\in \supp (\mathcal H)$ has at most two elements and thus can be represented as $\{s_1(r^*(P)),s_2(r^*(P))\}$ with $s_1(r^*(P))\leq r^*(P)\leq s_2(r^*(P))$. Moreover, for each $P,P'\in \supp(\mathcal H)$ with $r^*(P)<r^*(P')$, we have $s_2(r^*(P))\leq s_2(r^*(P'))$, as otherwise we would have $s_2(r^*(P'))\in (s_1(r^*(P)),s_2(r^*(P)))$ contradicting strict single-dippedness of $\mathcal H.$

Assume that there exists $P$ such that $s_1(r^*(P))<s_2(r^*(P))$, as otherwise the proposition obviously holds with $r^b=\ol s$. Define $r^b=\inf \{r^*(\tilde P):\tilde P\in \supp (\mathcal H),\ s_1(r^*(\tilde P))<s_2(r^*(\tilde P))\}$, so that, for each $P\in\supp (\mathcal H)$ with $r^*(P)<r^b$, we have $\supp (P)=\{r^*(P)\}$. Since $\supp(\mathcal H)$ is compact, there exists $P^b\in \supp(\mathcal H)$ with $r^*(P^b)=r^b.$ It follows that $\supp (P^b)=\{r^b\}$, as otherwise (i.e., if $s_1(r^*(P^b))<r^b<s_2(r^*(P^b))$ voter types in $(r^b,s_2(r^*(P^b))$ (which have strictly positive mass since $f$ is strictly positive on $[\ul s,\ol s]$) cannot be segregated, as this would contradict strict single-dippedness of $\mathcal H$, and also cannot be paired with other types, as this would contradict either strict single-dippedness of $\mathcal H$ or the definition of $r^b$.

Finally, we show that, for each $P,P'\in \supp(\mathcal H)$ with $r^b<r^*(P)<r^*(P')$, we have $s_1(r^*(P))\geq s_1(r^*(P'))$. Suppose by contradiction that $s_1(r^*(P))<s_1(r^*(P'))$. Since $\mathcal H$ is a strictly single-dipped pack-and-pair districting plan, by the definition of $r^b$, we have $s_1(r^*(P))<r^*(P)<s_2(r^*(P))\leq s_1(r^*(P'))<r^*(P')<s_2(r^*(P'))$. Define $r^\dagger=\inf \{r^*(\tilde P):\tilde P\in \supp (\mathcal H),\ s_1(r^*(P'))\leq s_1(r^*(\tilde P))<s_2(r^*(\tilde P))\leq s_2(r^*(P'))\}\geq s_1(r^*(P'))$. By the same argument as in the previous paragraph, we have $\delta_{r^\dagger} \in \supp(\mathcal H)$, contradicting that $\mathcal H$ is pack-and-pair.
\end{proof}

The next lemma restates some results from \citet{KCW}, which we use to prove Propositions \ref{p:segNAD} and \ref{p:nosegNAD}.
\begin{lemma}\label{l:snd} Consider the additive taste shock case where the taste shock density is strictly log-concave and symmetric about 0.
\begin{enumerate}
	\item If for all $s<r<s'$, we have
	\[
	G(r)< \frac{Q(s'-r)-\frac{1}{2}}{Q(s'-r)-Q(s-r)} G(s) +\frac{\frac{1}{2}-Q(s-r)}{Q(s'-r)-Q(s-r)}G(s'),
	\]
	then the unique optimal plan is segregation.
	\item If for all $s<s'$ there exists $r\in (s,s')$ such that
	\[
	G(r)> \frac{Q(s'-r)-\frac{1}{2}}{Q(s'-r)-Q(s-r)} G(s) +\frac{\frac{1}{2}-Q(s-r)}{Q(s'-r)-Q(s-r)}G(s'),
	\]
	then the unique optimal plan is negative assortative.
\end{enumerate}
\end{lemma}
\begin{proof}
By the definition of $r^*(P)$, we have 
\[r^*(\rho \delta_s+(1-\rho)\delta_{s'})=r\in (s,s') \iff  \rho=\frac{Q(s'-r)-\frac{1}{2}}{Q(s'-r)-Q(s-r)}\in (0,1).\]
Thus, part 1 says that, for any $s<s'$, the designer prefers to separate any district $P=\rho \delta_{s}+(1-\rho) \delta_{s'}$ into districts $\delta_s$ and $\delta_{s'}$, and part 2 says that, for any $s<s'$, the designer prefers to pool districts $\delta_s$ and $\delta_{s'}$ into some district $P=\rho \delta_{s}+(1-\rho) \delta_{s'}$. Consequently, parts 1 and 2 follow from Theorems 4 and 6 in \citet{KCW}.
\end{proof}

\begin{proof}[Proof of Proposition \ref{p:segNAD}]
For part 1, by Lemma \ref{l:snd}, negative assortative districting is uniquely optimal if for all $s<s'$ there exists $r\in (s,s')$ such that
\[
(G(r)-G(s))\left(Q(s'-r)-\tfrac 12\right) > (G(s')-G(r))\left(\tfrac 12- Q(s-r)\right),
\]
and thus, considering $r\uparrow s'$, if for all $s<s'$, we have
\[
(G(s')-G(s))q(0)>g(s')\left (\tfrac 12 - Q(s-s')\right),
\]
which  holds if $G$ is concave, as shown in the main text.

For part 2, it suffices to show that there exists $c>0$ such that, for all $s\neq r$, we have
\[
\frac{G(s)-G(r)}{g(r)}> \frac{Q(s-r)-\frac 1 2}{q(0)}.
\]
Indeed, this inequality implies that for all $s<r<s'$, we have
\begin{equation}\label{e:segregation}
\frac{G(r)-G(s)}{\frac 1 2-Q(s-r)}< \frac{g(r)}{q(0)}< \frac{G(s')-G(r)}{Q(s'-r)-\frac 1 2},	
\end{equation}
and hence segregation is uniquely optimal by Lemma \ref{l:snd}.

Now, since $g'(r)/g(r)\geq c$ for all $r$, Gronwall's inequality gives $g(s)/g(r)\geq e^{c(s-r)}$ for all $s>r$ and $g(s)/g(r)\leq e^{c(s-r)}$ for all $s<r$. Hence, for all $s,r$, we have
\[
\frac{G(s)-G(r)}{g(r)}=\int_r^s \frac{g(x)}{g(r)}d x\geq \int_r^s e^{c(x-r)} d x=\frac {e^{c(s-r)}-1}{c}.
\]
Thus, it suffices to show that there exists $c>0$ such that, for all $s\neq r$, we have
\[
\frac {e^{c(s-r)}-1}{c}> \frac{Q(s-r)-\frac 1 2}{q(0)}.
\]
Note that both sides have the same values and the same derivatives at $s=r$. Moreover, at $s=r$, the second derivative of the left-hand side, $c>0$, is greater than the second derivative of the right-hand side, ${q'(0)}/{q(0)}=0$. Thus, the inequality holds in some neighborhood $s\in (r-\varepsilon,r)$. Setting $c= q(0)/(1/2-Q(-\varepsilon))>0$ guarantees that the inequality holds for all $s\neq r$. Indeed, for $s\leq  r-\varepsilon$, we have
\[
\frac {e^{c(s-r)}-1}{c}> -\frac{1}{c}=\frac{Q(-\varepsilon)-\frac 12}{q(0)}\geq  \frac{Q(s-r)-\frac 1 2}{q(0)},
\]
where the first inequality holds by $e^{c(s-r)}>0$ and the second holds by monotonicity of $Q$. For $s> r$, we have
\[
\frac {e^{c(s-r)}-1}{c}> s-r> \frac{Q(s-r)-\frac 1 2}{q(0)},
\]
where the first inequality holds by strict convexity of $e^{cx}$ in $x$ and the second holds by strict concavity of $Q$ on $[0,+\infty)$.
\end{proof}

\begin{proof}[Proof of Proposition \ref{p:nosegNAD}]
Since density $q$ is symmetric about 0 and density $f$ is strictly positive on $[\ul s,\ol s]$, we have $\ul s <r^*(F)<\ol s$.
Since $G$ is strictly S-shaped with inflection point $r^*(F)$, it follows that $G$ is concave on $[r^*(F),\ol s]$. Thus, by Proposition \ref{p:segNAD}, negative assortative districting is uniquely optimal for types in $[r^*(F),\ol s]$, showing that segregation cannot be optimal.

Suppose for contradiction that negative assortative districting $\mathcal H$ is optimal. By Proposition \ref{p:pp}, for each $P\in \supp(\mathcal H)$ except for $\delta_{r^b}$, we have $s_1(r^*(P))<r^*(P)<s_2(r^*(P))$, where $s_1$ is decreasing and $s_2$ is increasing. Note that $r^b<r^*(F)$, because
\begin{gather*}
\int Q(s-r^*(F)) dF(s)=0=\iint Q(s-r^*(P))dP(s) d\mathcal H(P)\\
<\iint Q(s-r^b)dP(s) d\mathcal H(P)=\int Q(s-r^b) dF(s),
\end{gather*}
where the first two equalities hold by the definition of $r^*(F)$ and $r^*(P)$, the inequality holds by $r^*(P)>r^b$ for all $P\in \supp(\mathcal H)$ except for $\delta_{r^b}$, and the last equality holds by $\int P d\mathcal H(P)=F$. Since density $f$ is strictly positive on $[\ul s,\ol s]$, by the same argument as in the proof of Proposition \ref{p:pp}, we get $\lim_{r\downarrow r^b} s_1(r)=\lim_{r\downarrow r^b} s_2(r)=r^b$. Thus, for any $\varepsilon>0$, there exists $P\in \supp (\mathcal H)$ such that $r^b-\varepsilon<s_1(r^*(P))<s_2(r^*(P))<r^b+\varepsilon$, and all types in $[s_1(r^*(P)),s_2(r^*(P))]$ are matched between themselves in a negatively assortative manner. For small enough $\varepsilon>0$ and all $s<r<s'$ in $ [s_1(r^*(P)),s_2(r^*(P))]$, we have
\[
\frac{G(s)-G(r)}{g(r)} > \frac{Q(s-r)-\frac 12}{q(0)},
\]
where the inequality holds because both sides have the same values and the same derivatives at $s=r$, while the second derivative of the left-hand side, $g'(r)/g(r)>0$ (recall that $r^b$ is less than inflection point $r^*(F)$ of strictly S-shaped $G$), is greater than the second derivative of the right-hand side, $q'(0)/q(0)=0$. As follows from \eqref{e:segregation} in the proof of Proposition \ref{p:segNAD}, segregation is uniquely optimal for types in $[s_1(r^*(P)),s_2(r^*(P))]$, showing that $\mathcal H$ cannot be optimal.
\end{proof}

\begin{proof}[Proof of Proposition \ref{p:pap}] Suppose for contradiction that there exists an optimal non-pack-and-crack plan $\mathcal H$. By Proposition \ref{p:sdd}, $\mathcal H$ is strictly single-dipped. Consequently, since $\mathcal H$ is not pack-and-crack, there exist $s<r<s'\leq s''$ and $P,P'\in \supp (\mathcal H)$ such that $r^*(P)=r,$ $\supp(P)=\{s,s'\}$, and $\supp (P')=\{s''\}$. By Lemma \ref{l:dual}, condition \eqref{e:pap} holds. Intuitively, \eqref{e:pap} says that the designer prefers not to move a few type-$s$ voters from district $P$ to districts $\delta_{s}$ and $\delta_{s''}$.

We have numerically verified that \eqref{e:pap} holds over the specified range of parameters. The code is available on request.
\end{proof}

\begin{proof}[Proof of Proposition \ref{p:gamma}]
By Lemma \ref{l:dual}, $\lambda$ has a derivative $\lambda'(r)$ at each $r\in (r^b,r^b+\varepsilon]$ satisfying
\begin{align*}
g(r)-\lambda(r)q(s_2(r)-r)+\lambda'(r)\left(Q(s_2(r)-r)-\tfrac 12 \right)=0,\\
g(r)-\lambda(r)q(s_1(r)-r)+\lambda'(r)\left(Q(s_1(r)-r)-\tfrac 12 \right)=0.
\end{align*}
Solving for $\lambda(r)$ and $\lambda'(r)$ yields, for all $r\in (r^b,r^b+\varepsilon]$,
\begin{gather*}
\lambda(r) =\frac{g(r)[Q(s_2(r)-r)-Q(s_1(r)-r)]}{\left(Q(s_2(r)-r)-\frac 12\right) q(s_1(r)-r)-\left(Q(s_1(r)-r)-\frac 12\right)q(s_2(r)-r)},\\
\lambda'(r)=\frac{g(r)[q(s_2(r)-r)-q(s_1(r)-r)]}{\left(Q(s_2(r)-r)-\frac 12\right)q(s_1(r)-r)-\left(Q(s_1(r)-r)-\frac 12\right)q(s_2(r)-r)}.
\end{gather*}
Since $\lambda'$ is the derivative of $\lambda$, we have $d\lambda(r)/dr=\lambda'(r)$ for all $r\in (r^b,r^b+\varepsilon]$.
Taking into account that $s_1$ and $s_2$ are twice differentiable and satisfy $\lim_{r\downarrow r^b} s_1(r)=\lim_{r\downarrow r^b}s_2(r)=r^b$, we can apply L'Hopital's rule to evaluate $d\lambda(r)/dr=\lambda'(r)$ in the limit $r\downarrow r^b$ to obtain
\[
\frac{g'(r^b)q(0)}{(q(0))^2}=\frac{g(r^b)q'(0)}{(q(0))^2},
\]
which implies that $r^b=0$, because $G(r)=Q(\gamma r)$ for all $r$ and $q'(r)=0$ iff $r=0$. Denote $\lim_{r\downarrow r^b} s'_1(r)=1-\beta_1$ and $\lim_{r\downarrow r^b} s'_2(r)=1+\beta_2$, where $\beta_1\geq 1$ (because $s_1$ is decreasing) and $\beta_2\geq 0$ (because $s_2(r)>r)$.
Differentiating $d\lambda(r)/dr=\lambda'(r)$ with respect to $r$ and taking the limit $r\downarrow 0$, we get
\[
\frac{\gamma q''(0)(\gamma^2-\beta_2\beta_1)}{q(0)}=\frac{\gamma q''(0)(\beta_2-\beta_1)}{2q(0)},
\]
and hence
\begin{equation}\label{e:Y1}
2\gamma^2=2\beta_2\beta_1+\beta_2-\beta_1.	
\end{equation}
Since, for small enough $r>0$, type $s_1(r)$ is assigned to both district $\delta_{s_1(r)}$ and district $P$ with $r^*(P)=r$ and $\supp(P)=\{s_1(r),s_2(r)\}$, we must have, by Lemma \ref{l:dual},
\begin{equation*}\label{e:ps1r}
	Q(\gamma s_1(r))=Q(\gamma r) +\lambda(r)\left(Q(s_1(r)-r)-\tfrac 12\right).
\end{equation*}
In the limit $r\downarrow 0$, the values and the derivatives up to order 2 of both sides always coincide, while the third derivatives coincide iff
\begin{gather*}
q''(0)\gamma^3(-\beta_1+1)^3=q''(0)\gamma^3-3q''(0)\gamma^3\beta_1+3q''(0)\gamma\beta_2\beta_1^2-q''(0)\gamma\beta_1^3,
\end{gather*}
which simplifies to
\begin{equation}\label{e:Y2}
-\gamma^2\beta_1+3\gamma^2=3\beta_2-\beta_1.	
\end{equation}
Since, for small enough $r>0$, type $s_1(r)$ is assigned to both district $\delta_{s_1(r)}$ and district $P$ with $r^*(P)=r$, while type $s_2(r)$ is assigned only to district $P$, we have
\[
f(s_1(r))s_1'(r)\left(Q(s_1(r)-r)-\tfrac 12\right) \geq f(s_2(r))s_2'(r)\left(Q(s_2(r)-r)-\tfrac 12\right).
\]
In the limit $r\downarrow 0$, both sides are equal, and hence their derivatives must satisfy
\[
-f(0)q(0)\beta_1 (1-\beta_1)\geq f(0)q(0)\beta_2 (\beta_2+1),
\]
which, given that $\beta_1+\beta_2>0$, simplifies to
\begin{equation}\label{e:Y3}
\beta_1\geq \beta_2+1.	
\end{equation}
Equations \eqref{e:Y1} and \eqref{e:Y2} have two solutions $(\beta_1,\beta_2)=\left({3\gamma^2}/{(2(\gamma^2-1))},{\gamma^2}/{2}\right)$ and $(\beta_1,\beta_2)=\left(1,{(2\gamma^2+1)}/{3}\right)$, unless  $\gamma^2=1$, in which case \eqref{e:Y1} and \eqref{e:Y2} have only one solution $(\beta_1,\beta_2)=(1,1)$. The solution $(\beta_1,\beta_2)=\left(1,{(2\gamma^2+1)}/{3}\right)$ never satisfies \eqref{e:Y3} and thus is discarded. Moreover, for the solution $(\beta_1,\beta_2)=\left({3\gamma^2}/{(2(\gamma^2-1))},{\gamma^2}/{2}\right)$, condition $\beta_1\geq 1$ yields $\gamma>1$, and condition \eqref{e:Y3} yields $\gamma\leq \sqrt{1+\sqrt 3}$. Thus, for Y-districting to be optimal, we must have $\gamma\in (1,\sqrt{1+\sqrt 3} ]$. Finally, the statement in Footnote \ref{f:gamma} holds because
\[
\lim_{r\downarrow 0} s_1'(r)=1-\beta_1=-\frac{(\gamma^2+2)}{2(\gamma^2-1)}<0\quad\text{and}\quad \lim_{r\downarrow 0} s_2'(r)=1+\beta_2=1+\frac{\gamma^2}{2}>0
\]
are both strictly increasing in $\gamma$.
\end{proof}

\end{document}